\documentclass[final]{siamart190516}

\usepackage{amsfonts}
\usepackage{graphicx}
\usepackage{epstopdf}
\usepackage{algorithmic}
\ifpdf
  \DeclareGraphicsExtensions{.eps,.pdf,.png,.jpg}
\else
  \DeclareGraphicsExtensions{.eps}
\fi

\newsiamremark{remark}{Remark}
\newsiamremark{hypothesis}{Hypothesis}
\crefname{hypothesis}{Hypothesis}{Hypotheses}
\newsiamthm{claim}{Claim}
\newsiamthm{assumption}{Assumption}
\newsiamremark{example}{Example}

\headers{Risk sensitive portfolio optimization}{M. Pitera and {\L}. Stettner}

\title{Discrete-time risk sensitive portfolio optimization with proportional transaction costs}

\author{Marcin Pitera\thanks{Institute of Mathematics, Jagiellonian University, Krakow, Poland,
  (\email{marcin.pitera@uj.edu.pl});  research supported by NCN grant 2020/37/B/ST1/00463.} \and {\L}ukasz Stettner\thanks{Institute of Mathematics, Polish Academy of Sciences, Warsaw, Poland,
  (\email{l.stettner@impan.pl});  research supported by NCN grant 2020/37/B/ST1/00463.}}

\usepackage{amsopn}

\usepackage{enumitem} 

\def\cS{\mathcal{S}}
\def\cB{\mathcal{B}}

\def\bP{\mathbb{P}}
\def\bH{\mathbb{H}}
\def\bE{\mathbb{E}}
\def\bR{\mathbb{R}}
\def\bT{\mathbb{T}}
\def\bN{\mathbb{N}}

\makeatletter
\def\namedlabel#1#2{\begingroup
    #2%
    \def\@currentlabel{#2}%
    \phantomsection\label{#1}\endgroup
}
\makeatother

\begin{document}

\maketitle

\begin{abstract}
In this paper we consider a discrete-time risk sensitive portfolio optimization over a long time horizon with proportional transaction costs. We show that within the log-return i.i.d. framework  the solution to a suitable Bellman equation exists under minimal assumptions and can be used to characterize the optimal strategies for both risk-averse and risk-seeking cases. Moreover, using numerical examples, we show how a Bellman equation analysis can be used to construct or refine optimal trading strategies in the presence of transaction costs.
\end{abstract}



\section{Introduction}
Quantitative portfolio management is an important part of mathematical finance. Stimulated by the seminal work~\cite{Mar1952}, this field has been consistently evolving during the last 70 years for both discrete and continuous time settings, see \cite{KolTutFab2014,Cha2017,Pri2007} and references therein for an overview. Among the considered portfolio optimisation frameworks, the risk sensitive portfolio optimisation is among the most recognised ones, see \cite{DavLle2014,BiePli1999,FleShe2000}. Given a wealth process $(W_t)$ and the risk-averse parameter $\gamma\neq 0$, the long-run version of risk sensitive criterion is defined as
\begin{equation}\label{eq:intro1}
\liminf_{t\to \infty}\tfrac{1}{t}\tfrac{1}{\gamma}\bE\left[ W_t^{\gamma}\right]\quad \textrm{or equivalently as }\quad \liminf_{t\to\infty}\tfrac{1}{t}\mu^{\gamma}(\ln W_t),
\end{equation}
where $\mu^{\gamma}(\cdot):=1/\gamma \cdot \bE[\exp(\gamma(\cdot))]$ is the entropic utility. The optimality criterion presented in \eqref{eq:intro1} measures the long-run normalised entropy of the log-wealth and could be seen as a non-linear extension of the {\it Kelly's criterion}, see~\cite{MacThoZie2011,BamVicVic2002,DavLlo2021}. It should be noted that the extension~\eqref{eq:intro1} is in fact unique within the class of cash-additive and strongly time-consistent certainty equivalents which explains why the usage of entropy is so common in multiple stochastic control applications, see~\cite{KupSch2009,CheMas2004} for details. In fact, the risk sensitive criterion appears naturally in many portfolio investment problems and is linked to various optimality frameworks. For completeness, let us provide some examples. First, by considering the second order Taylor's expansion of the entropic utility, around $\gamma=0$, we get 
\[
\mu^{\gamma}(\ln W_t)= \bE[\ln W_t] +\tfrac{\gamma}{2}\textrm{Var}[\ln W_t] +O(\gamma_2,t),\quad t>0,
\]
which shows that, for $\gamma<0$, the risk sensitive framework might be considered as an extension of the mean-variance Markowitz portfolio optimisation that allows time-consistent utility treatment, see~\cite{BiePli2003}. Second, \eqref{eq:intro1} is directly linked to so called {\it equivalent safe rate}, which reflects the minimal hypothetical safe rate that would encourage the investor to invest in the risky portfolio, see~\cite{GuaTolWan2019}. Third, for the risk-averse case $\gamma <0$, the risk-sensitive criterion is dual to the {\it downside risk}, which is a common investment criterion in the long-run portfolio optimisation, see \cite{Nag2012} or \cite{Pha2015} for details. Fourth, for $\gamma > 0$, the maximization of \eqref{eq:intro1} is related to the studies of {\it power utility asymptotics} and can be considered as a dual problem to {\it upside chance probability}, see  \cite{Pha2015} and \cite{Ste2011b}. Finally, let us remark that risk sensitive criterion is an {\it acceptability indicex} (also called {\it performance measure}) and has many economically desirable properties, see e.g.~\cite{CheMad2009,BieCiaDraKar2016}. We refer to~\cite{BiePli2003} for an overview of economic properties of risk sensitive criterion made in reference to portfolio management.

The main aim of this paper is to show that under the i.i.d. property imposed on asset's log-returns one can solve a suitable risk-sensitive Bellman equation under proportional transaction cost and lack of short selling; see  \cite{DacNor1990,CziSch2016} for a discussion about transaction cost impact on portfolio management. We emphasize that the set of additional assumptions imposed on log-returns in this paper is minimal, i.e. we only require that asset's log-returns have finite mean and entropy. While this might be counter-intuitive on the first sight, as one typically impose strong ergodic assumption on the process in order to get the existence of risk-sensitive Bellman equation solution, the i.i.d. property proves to be a plausible alternative. For an overview of the key results, we refer to Theorem~\ref{th:main1}, Theorem~\ref{th:main2}, and Theorem~\ref{th:main3}.

The results of this paper are presented in a self-contained entropy based way to streamline the economic context; we hope this makes the paper more transparent and accessible to the generic mathematical finance community. That saying, the results presented here are in fact linked to an extensive literature on the risk sensitive stochastic control optimisation, and are expanding this framework in reference to portfolio management, see \cite{BiePli1999,PitSte2016}. That is why in some cases we decided to present alternative formulations of the Bellman equations, to link them more directly with the  (controlled) {\it Multiplicative Poisson Equation} framework.

Our work is also linked to a variety of problems studied for Markov decision processes (see e.g. \cite{BauJas2018}) and recently studied continuous time risk sensitive problems with regime switching over finite time horizon, see \cite{BoLiaYu2019, Hat2018, DasGosRan2018}. In particular, we want to mention that the proof techniques presented in this paper are based on some novel ideas applied to {\it vanishing-discount} and {\it span-contraction} approaches, cf. \cite{CavHer2017,SheStaObe2013}. For instance, we were able to weaken the typical assumption imposed on the negative value of the log-process, by replacing Schwarz's inequality based approximation with a tail-based argument in one of the key steps of local contraction property proof, see Proposition~\ref{eq:span.contr.bounded}. Also, by incorportaing Arzela-Ascoli theorem into the vanishing discount approach, we were able to show the existence of a regular Bellman solution.

Apart from theoretical results, we present two numerical examples. They might be interesting to a reader who is not familiar with the risk sensitive stochastic control but wants to better understand why the study of Bellman equation could improve trading performance even in a very simplistic case. In particular, by using simple approximation schemes, one can directly recover no-action strategies that are important aspect of portfolio management in the presence of transaction costs, see \cite{CziSch2016}. This shows why the development of efficient risk sensitive approximation algorithms in the dynamic context can help to develop or benchmark trading strategies; see \cite{FeiYanWan2021,BasBhaBor2008,Bor2010,AraBisPra2021} where practical aspects linked to risk sensitive policy iteration algorithms are studied.

This paper is organized as follows. In Section~\ref{S:introduction}, we provide the general setup, state the assumptions, and formulate suitable Bellman equations. Next, in Section~\ref{S:discounted} we focus on the discounted version of the problem, that paves the ground for the usage of the vanishing discount approach. In Section~\ref{S:vanishing} we follow the vanishing discount approach in order to show the key results of this paper. Then, in Section~\ref{S:span-contraction} we switch to the span-contraction approach in order to strengthen the results presented in Theorem~\ref{th:main2} and show how to utilise the local contraction property in the i.i.d. setting. Finally, in Section~\ref{S:examples} we present numerical examples.

\section{Problem formulation}\label{S:introduction}

Let $(\Omega, F,(F_t)_{t\in\bT},\bP)$ be a discrete-time filtered probability space, where $\bT=\bN$. Let $d\in\bN$ denote the number of available risky assets and let 
$S(t):=(S_1(t),\ldots,S_d(t))$ denote the positive vector price process, where $S_j(t)$ denotes the price of the $j$th risky asset at time $t\in\bT$. For a given trading strategy, we use $N(t)=(N_1(t),\ldots,N_d(t))$ to denote the portfolio asset volume vector at time $t$ after the portfolio is rebalanced, i.e. $N_i(t)$ denotes how much asset $S_i$ we hold in our portfolio at time $t$ after the rebalancing is executed. Also, we use
\begin{equation}\label{eq:Wt}
W(t-):= \langle N(t-),S(t)\rangle\quad\textrm{and}\quad W(t):=\langle N(t),S(t)\rangle,
\end{equation}
where $\langle\cdot,\cdot\rangle$ is the standard scalar product, to denote portfolio wealth process at time $t$ before and after the rebalancing, respectively. Throughout the paper we assume absence of short selling and follow the proportional transaction cost framework. This is partly encoded in the self-financing condition that is given by
\begin{equation}\label{eq:self-financing}
W(t)=W(t-)-d( (N(t)-N(t-1))\cdot S(t)),
\end{equation}
where $(\cdot)$ denotes vector point-wise product, function $d\colon \bR^d\to\bR_{+}$ is the proportional transaction cost penalty given by
\[
d(x):=\langle c, [x]^+\rangle+ \langle h, [x]^-\rangle,\quad x\in\bR^d,
\]
for a fixed cost rates $c,h\in \bR^d$ such that $0<c_j, h_j<1$, $j=1,2,\ldots,d$, and $[x]^\pm$ denotes (component-wise) positive/negative part of $x$. To ease the notation, we introduce portfolio loading factors (portions of the capital invested in the assets, weights) vectors given by
\begin{equation}\label{eq:pi}
\pi(t):=\frac{N(t)\cdot S(t)}{W(t)}\quad\textrm{and}\quad \pi(t-):=\frac{N(t-1)\cdot S(t)}{W(t-)}.
\end{equation}
Note that due to the absence of short selling, for any $t\in\bT$ and $\omega\in\Omega$, we have $\pi(t)(\omega),\pi(t-)(\omega)\in {\cal S}$, where ${\cal{S}}:=\left\{x\in \bR^d: x_j\geq 0; \langle x,1\rangle=1\right\}$. Before we introduce the objective function, let us present a short technical lemma which shows that capital decay $W(t)/W(t-)$ can be expressed as a function of factor loadings.

\begin{lemma}\label{lemma:WtoPi}
There is $\tilde s\in (0,1)$ and continuous function $s: {\cal S}^2\mapsto [\tilde s,1]$ such that
\[
\frac{W(t)}{W(t-)}=s(\pi(t-),\pi(t)),\quad t\in\bT.
\]
\end{lemma}

\begin{proof}
Let $F\colon \bR_+\times {\cal S}^2\to \bR$ be a function given by $F(w,x,y):=w + d (wy-x)$ and let $\tilde e:=\min_{j=1,\ldots,d}h_j$. Noting that $F$ is continuous, strictly increasing in $w$, $\tilde e \leq F(0,x,y)\leq \max_{j=1,\ldots,d}h_j<1$, and $F(1,x,y)>1$, we know that there exists function $s\colon {\cal S}^2\to [\tilde s,1]$ such that $F\left(s(\pi(t-),\pi(t)),\pi(t-),\pi(t)\right)=1$. On the other hand, using self-financing condition \eqref{eq:self-financing} we get
\[
 F\left(\tfrac{W(t)}{W(t-)},\pi(t-),\pi(t)\right)=\tfrac{W(t)}{W(t-)} + d \left(\tfrac{W(t)}{W(t-)}\pi(t)-\pi(t-)\right)=1.
\]
Since $F$ is strictly increasing with respect to $w$ we know that $s(\pi(t-),\pi(t))=\tfrac{W(t)}{W(t-)}$. It remains to show that $s$ is continuous. Let ${\cal S}\ni\pi_n, \pi_n'$ be such that $\pi_n\to \pi$ and $\pi_n'\to \pi'$, as $n\to \infty$. Recalling that $F$ is continuous, strictly increasing in its first argument and satisfies $F(s(\pi_n,\pi_n'),\pi_n,\pi_n')=1$ as well as $F(s(\pi,\pi'),\pi,\pi')=1$, we conclude that for any subsequence $(n_k)_{k\in\bN}$ such that $s(\pi_{n_k},\pi_{n_k}')\to \bar{s}$, for some $\bar{s} \in [\tilde s,1]$, we get $\bar{s}=s(\pi,\pi')$. Since the same limit $s(\pi,\pi')$ is achieved for any subsequence $(n_k)$, we get continuity of $s$.
\end{proof}
From Lemma \ref{lemma:WtoPi} we see that the trading strategy could be represented via the loading factors \eqref{eq:pi}. For any given ${\cal{S}}$-valued (adapted) strategy $\pi$ we use $W_{\pi}$ to denote the corresponding wealth process. 

The main goal of this paper is to find strategy $\pi$ that maximizes long run risk sensitive objective function,  applied to log-wealth process. Namely, we fix a risk-sensitivity parameter $\gamma\in \bR\setminus \{0\}$ and consider the objective function given by
\begin{align}
J(\pi) &:=\liminf_{n\to \infty}\frac{1}{n}\frac{1}{\gamma}\bE\left[ W_{\pi}(n-)^{\gamma}\right]\nonumber\\
& =\liminf_{n\to \infty}\frac{\mu^{\gamma}(\ln W_{\pi}(n-))}{n}\nonumber\\
& =\liminf_{n\to \infty}\frac{1}{n}\mu^{\gamma}\left(\sum_{t=0}^{n-1}\ln \frac{W_{\pi}((t+1)-)}{W_{\pi}(t-)}\right),\label{eq:obj1}
\end{align}
where $\mu^{\gamma}(X):=\frac{1}{\gamma}\bE[e^{\gamma X}]$ is the entropic utility function; for consistency, we also use limit notation $\mu^{0}(X):=\bE[X]$. Note that \eqref{eq:obj1} is measuring time averaged entropy of portfolio's log-return; see \cite{BiePli2003} for the economical context.

Since we are interested in optimising portfolio's log-growth, throughout this paper we assume that the assets log-return vector $r(t)=(r_i(t))_{i=1}^{d}$, where $r_i(t):=\ln \tfrac{S_i(t)}{S_i(t-1)}$, is an i.i.d. vector satisfying conditions
\begin{equation}\label{A.1} \tag{A.1}
\mu^{\gamma}\left( r_i(t)\right)\in\bR\quad\textrm{and}\quad \bE\left[ r_i(t)\right]\in\bR\quad\textrm{for } i=1,2,\ldots,d,
\end{equation}
which means that log-returns are integrable and have finite entropy for the prefixed risk sensitive parameter $\gamma\in \bR\setminus\{0\}$.

\begin{remark}
From assumption \eqref{A.1}, using monotonicity of entropic risk with respect to risk-sensitivity parameter, we get that for any $\delta$ between $\gamma$ and $0$, we have $\mu^{\delta}(r_i(t))\in\bR$. Also, note that assumption \eqref{A.1} could be rephrased using non-entropy notation as $\bE[e^{\gamma r_i(t)}]\in \bR$ and $\bE[r_i(t)]\in\bR$, which could be linked to log-returns moment generating function finiteness.
\end{remark}
For transparency, we also introduce an asset relative shift process $w(t)=(w_i(t))_{i=1}^{d}$ given by
\[
w(t):=e^{r(t)}=\left(\frac{S_1(t)}{S_1(t-1)},\ldots,\frac{S_d(t)}{S_d(t-1)}\right).
\]
Let us now show how to re-express inner part of $J(\pi)$ as a $\pi$-controlled process. Let $G\colon \bR^d\times \bR^d\to\bR^d$ be given by $G(x,y):=x\cdot y/\langle x,y\rangle$. Noting that
\[
\pi(t-)=\frac{\pi(t-1)\cdot w(t)}{\langle \pi(t-1), w(t)\rangle}=G(\pi(t-1),w(t)),
\]
and rewriting the objective criterion \eqref{eq:obj1} as
\begin{align}
J(\pi)& =\liminf_{n\to \infty}\frac{1}{n}\mu^{\gamma}\left(\sum_{t=0}^{n-1}\ln \frac{W_{\pi}((t+1)-)}{W_{\pi}(t)}\frac{W_{\pi}(t)}{W_{\pi}(t-)}\right)\nonumber\\
& =\liminf_{n\to \infty}\frac{1}{n}\mu^{\gamma}\left(\sum_{t=0}^{n-1}\ln \frac{\langle N_{\pi}(t),S(t+1)\rangle}{W_{\pi}(t)}+\ln(s(\pi(t-),\pi(t)))\right)\nonumber\\
& =\liminf_{n\to \infty}\frac{1}{n}\mu^{\gamma}\left(\sum_{t=0}^{n-1}\ln \frac{\langle N_{\pi}(t),w(t+1)\cdot S(t)\rangle}{W_{\pi}(t)}+\ln(s(\pi(t-),\pi(t)))\right)\nonumber\\
& =\liminf_{n\to \infty}\frac{1}{n}\mu^{\gamma}\left(\sum_{t=0}^{n-1} \ln \langle \pi(t), w(t+1)\rangle+\ln(s(G(\pi(t-1),w(t)),\pi(t)))\right), \label{eq:obj2}
\end{align}
we  essentially get direct restatement of the controlled log-wealth process using control process $\pi$ and independent shifts $w$. 

In order to solve \eqref{eq:obj2}, we introduce the associated Bellman equation. Ideally, given $\gamma\in \bR_{*}:=\bR\setminus\{0\}$, we are looking for a function $v: {\cal S}\to \bR$ and a constant $\lambda\in\bR$ which satisfy equation
\begin{equation}\label{eq:Bellman0}
\lambda+v(\pi)=\sup_{\pi'\in \mathcal S}\left[\mu^{\gamma} \Big( \ln\langle \pi', w(1)\rangle+\ln s(\pi,\pi')+v(G(\pi',w(1)))\Big)\right],
\end{equation}
for $\pi\in \cS$. Note that, with a slight abuse of notation, in \eqref{eq:Bellman0} we use $\pi,\pi'\in \mathcal S$ in reference to a deterministic pre-rebalancing and post-rebalancing weights rather than the whole rebalancing strategy; this convention is often used in the paper when optimality equations are considered.

For technical reasons, instead of considering \eqref{eq:Bellman0} directly, in this paper we consider its slightly modified version given by 
\begin{equation}\label{eq:Bellman0.2}
\lambda+v(\pi,\gamma)=\gamma\sup_{\pi'\in \mathcal S}\left[ \mu^{\gamma} \Big( \ln\langle \pi', w(1)\rangle+\ln s(\pi,\pi')+\gamma^{-1}v(G(\pi',w(1)),\gamma)\Big)\right].
\end{equation}
where $v: {\cal S}\times\bR_{*}\to \bR$ and $\lambda\in\bR$. In particular, we use $v(\pi,\gamma)$ instead of $v(\pi)$  to emphasize the dependency between risk-parameter choice and Bellman's equation solution, and to embed \eqref{eq:Bellman0} into vanishing discount framework. Note that if $v(\cdot,\gamma)$ and $\lambda$ solves \eqref{eq:Bellman0.2}, then $\gamma^{-1} v(\cdot,\gamma)$ and $\gamma^{-1}\lambda$ solves \eqref{eq:Bellman0}. Also, Bellman's equation \eqref{eq:Bellman0.2} could be directly restated in a more classical form. Namely, for the risk-averse case $\gamma<0$ we can rephrase \eqref{eq:Bellman0.2} as
\begin{equation}\label{eq:Bellman1}
e^{v(\pi,\gamma)}=\inf_{\pi'\in\mathcal S}e^{\gamma \ln s(\pi,\pi')}\bE \left[e^{\gamma \left(\ln(\langle \pi', w(1)\rangle)-\lambda\right)+v(G(\pi',w(1)),\gamma)}\right].
\end{equation}
while for the risk-seeking case $\gamma>0$ we get
\begin{equation}\label{eq:Bellman1b}
e^{v(\pi,\gamma)}=\sup_{\pi'\in\mathcal S}e^{\gamma \ln s(\pi,\pi')}\bE \left[e^{\gamma \left(\ln(\langle \pi', w(1)\rangle)-\lambda\right)+v(G(\pi',w(1)),\gamma)}\right].
\end{equation}
For transparency, if not stated otherwise, we use $v$ in reference to equation \eqref{eq:Bellman0.2}.

In the following sections, we study the existence of the solution to the initial problem and its link to Bellman's equation \eqref{eq:Bellman0}. For risk-averse case $\gamma<0$, we will show that under general assumption \eqref{A.1} and can solve the recursive version of \eqref{eq:Bellman0} in order to get the optimal constant and optimal strategy. Moreover, by imposing additional condition on $w$ that is related to mixing, one can show that \eqref{eq:Bellman0} could be directly solved without relying on the recursive scheme. On the other hand, for risk-seeking case $\gamma>0$ assumption \eqref{A.1} alone imply existence of solution to \eqref{eq:Bellman0}. 

For completeness, we will also show that selectors to the Bellman equation determine an optimal strategy in both cases. As already said, the results will be obtained using both vanishing discount approach as well as span-contraction approach.


\section{Discounted problem}\label{S:discounted}
Before we apply the vanishing discount approach, let us provide a few remarks for the discounted version of problem \eqref{eq:obj1}. Consider $\alpha>0$ and the discounted risk sensitive objective problem given by
\[
\sup_{\pi}\tilde J_{\alpha}(\pi) =\sup_{\pi} \mu^{\gamma}\left(\sum_{t=0}^\infty e^{-\alpha t}\left[\ln (\langle \pi(t), w(t+1)\rangle)+\ln(s(\pi(t-),\pi(t))\right]\right).
\]
The associated discounted analogue of Bellman equation \eqref{eq:Bellman0.2} is given by
\begin{equation}\label{eq:Bellman1.d}
v_{\alpha}(\pi,\gamma)=\gamma\sup_{\pi'\in \mathcal S}\left[ \mu^{\gamma} \Big( \ln\langle \pi', w(1)\rangle+\ln s(\pi,\pi')+\gamma^{-1}v_{\alpha}(G(\pi',w(1)),\gamma e^{-\alpha})\Big)\right].
\end{equation}
Note that \eqref{eq:Bellman1.d} is in fact linked to a series of equations which effectively should provide the formula for $(v_{\alpha}(\pi,\gamma e^{-n\alpha}))_{n\in\bN}$. Nevertheless, for simplicity, we are looking for a stronger condition, i.e. a function $v_\alpha$ that satisfies \eqref{eq:Bellman1.d} for any value of risk sensitive parameter between $\gamma$ and 0. In other words, we want \eqref{eq:Bellman1.d} to hold on $\cS\times \Gamma$, for
\[
\Gamma:=[\gamma_-,\gamma_+]\setminus\{0\},
\]
where $\gamma_{-}:=\min\{0,\gamma\}$ and $\gamma_{+}:=\max\{0,\gamma\}$. Nevertheless, with slight abuse of notation, if no ambiguity arise, we often use $\gamma$ to denote a generic choice from $\Gamma$. 

As before, Equation \eqref{eq:Bellman1.d} could be rephrased in a classical way, i.e. for $\gamma<0$ we can restate \eqref{eq:Bellman1.d} as
\[
e^{v_{\alpha}(\pi,\gamma)}=\inf_{\pi'\in\mathcal S}e^{\gamma \ln s(\pi,\pi')}\bE \left[e^{\gamma\ln(\langle \pi', w(1)\rangle)+v_{\alpha}(G(\pi',w(1)),\gamma e^{-\alpha})}\right]
\]
while for $\gamma>0$ we can rewrite  \eqref{eq:Bellman1.d} as
\[
e^{v_{\alpha}(\pi,\gamma)}=\sup_{\pi'\in\mathcal S}e^{\gamma \ln s(\pi,\pi')}\bE \left[e^{\gamma\ln(\langle \pi', w(1)\rangle)+v_{\alpha}(G(\pi',w(1)),\gamma e^{-\alpha})}\right].
\]
Let us now introduce a lemma that will be helpful for establishing existence of solutions to Bellman's equation \eqref{eq:Bellman0.2}. For brevity, we introduce supplementary notation
\begin{align*}
z_{\gamma}(\pi',\pi)&:=\mu^{\gamma}\left(\ln \langle\pi,w(1)\rangle +\ln s(\pi',\pi)\right),\\
z^{-} &:=-\left|\min_i\mu^{\gamma}(r_i(1))\right|-d\max_i\bE|r_i(1)|-\tfrac{\ln d}{|\gamma|}+\ln \tilde s,\\
z^{+} &:=\phantom{-}\left|\max_i\mu^{\gamma}(r_i(1))\right|+d\max_i\bE|r_i(1)|+\tfrac{\ln d}{|\gamma|}.
\end{align*}
Note that $z^{-}$ and $z^{+}$ are finite due to \eqref{A.1}.
\begin{lemma}\label{lm:z.bounds}
Let $\gamma\in\bR\setminus\{0\}$ and let us assume \eqref{A.1}. Then, for any $\delta \in [\gamma_{-},\gamma_{+}]$ and $\pi,\pi'\in\cS$ we have $z^{-}\leq z_{\delta}(\pi',\pi)\leq z^{+}$.
\end{lemma}
\begin{proof}
First let us consider the case where $\gamma<0$ and $\delta\in [\gamma,0]$. Using monotonicity and translation invariance of entropic utility we get
\begin{align}
z_{\delta}(\pi',\pi) &\geq \mu^{\delta}\left(\ln \langle \pi,w(1)\rangle)\right)+\ln \tilde s\nonumber\\
&\textstyle\geq \mu^{\gamma}\left(\min_i r_i(1)\right)+\ln \tilde s\nonumber\\
 &\textstyle =\tfrac{1}{\gamma}\ln\bE\left[\max_i e^{\gamma  r_i(1)}\right]+\ln \tilde s\nonumber\\
&\textstyle\geq \tfrac{1}{\gamma}\ln\left(\sum_{i=1}^{d}\bE\left[e^{\gamma  r_i(1)}\right]\right)+\ln \tilde s\nonumber\\
&\textstyle\geq \min_i \frac{1}{\gamma}\ln \left[d\cdot \bE\left[e^{\gamma  r_i(1)}\right]\right)+\ln \tilde s\nonumber\\
&\textstyle \geq -\left|\min_i\mu^{\gamma}(r_i(1))\right|-\frac{\ln d}{|\gamma|}+ \ln \tilde s\label{eq:z.minus}
\end{align}
and
\begin{align}
\textstyle z_{\delta}(\pi',\pi)& \textstyle\leq \mu^{\delta}\left(\ln \langle \pi,w(1)\rangle)\right)+\ln 1 \nonumber\\
& \textstyle\leq \mu^{0}\left(\max_i r_i(1)\right)\nonumber\\
&\textstyle\leq d\max_i\bE|r_i(1)| \label{eq:z.plus}.
\end{align}
Now, for $\gamma>0$ and $\delta\in [0,\gamma]$, using similar calculations, we get
\begin{align}
\textstyle z_{\delta}(\pi',\pi)& \textstyle\geq -d\max_i\bE|r_i(1)|+\ln \tilde s, \nonumber\\
\textstyle z_{\delta}(\pi',\pi) &\textstyle\leq \left|\max_i\mu^{\gamma}(r_i(1))\right|+\frac{\ln d}{|\gamma|}.\label{eq:z.all}
\end{align}
Combining \eqref{eq:z.minus}, \eqref{eq:z.plus}, and \eqref{eq:z.all} we conclude the proof.

\end{proof}
We are now ready to present the main theorem of this section.

\begin{theorem}\label{th:discounted}
Let $\gamma\in\bR\setminus\{0\}$ and let us assume \eqref{A.1}. Then, for each $\alpha>0$, there exists a continuous and bounded function $v_{\alpha}\colon \cS\times \Gamma\to \bR$, that is a solution to discounted Bellman's equation \eqref{eq:Bellman1.d}.
\end{theorem}
\begin{proof}
For brevity we only show the proof for $\gamma<0$; the proof for $\gamma>0$ is analogous. Fix $\alpha>0$ and consider the Bellman operator linked to \eqref{eq:Bellman1.d} that is given by
\[
T v(\pi,\gamma):=\gamma\sup_{\pi'\in\cS}\left[\mu^{\gamma}\left(\ln\langle\pi',w(1)\rangle+\ln s(\pi,\pi')+\gamma^{-1}v(G(\pi',w(1)),\gamma e^{-\alpha})\right)\right],
\]
for any $v\colon \cS\times \Gamma \to \bR$. First, let us show that $T$ is $C$-Feller, i.e.  it transforms bounded and continuous functions into themselves. Using Lemma~\ref{lm:z.bounds} we immediately get
\[
\|Tv\|_{\textrm{sup}}\leq |\gamma|\cdot [z^+-z^{-}]+\|v\|_{\textrm{sup}},
\]
where $\|\cdot\|$ denotes the standard supremum norm, which shows that boundedness is preserved. Now, let us show that continuity is also preserved for continuous bounded functions. First, since $v$ is continuous, for any  sequence $((\pi'_n,\gamma_n))_{n\in\bN}$, where $\pi'_n\in \cS$ and $\gamma_n\in\Gamma$, satisfying $(\pi'_n,\gamma_n)\to(\pi',\gamma)$, we get
\[
e^{\gamma_n \ln\langle\pi_n',w(1)\rangle+v(G(\pi_n',w(1)),\gamma_n e^{-\alpha}))}\stackrel{a.s.}\longrightarrow e^{\gamma \ln\langle\pi',w(1)\rangle+v(G(\pi',w(1)),\gamma e^{-\alpha}))}.
\]
Now, using similar reasoning as in Lemma~\ref{lm:z.bounds} we know that 
\[
e^{\gamma_n \ln\langle\pi_n',w(1)\rangle+v(G(\pi_n',w(1)),\gamma_n e^{-\alpha}))}\leq e^{\gamma \min_i r_i(1)+\|v\|_{\textrm{sup}}}.
\]
Thus, noting that $e^{\gamma \min_i r_i(1)+\|v\|_{\textrm{sup}}}\in L^1$ due to \eqref{A.1}, and using dominated convergence theorem, we get
\[
\bE\left[e^{\gamma_n \ln\langle\pi_n',w(1)\rangle+v(G(\pi_n',w(1)),\gamma_n e^{-\alpha}))}\right] \longrightarrow \bE\left[e^{\gamma \ln\langle\pi',w(1)\rangle+v(G(\pi',w(1)),\gamma e^{-\alpha}))}\right],
\]
which in turn implies continuity of the mapping
\begin{equation}\label{eq:cont.1}
(\pi,\pi',\gamma)\to \gamma\left[\mu^{\gamma}\left(\ln\langle\pi',w(1)\rangle+\ln s(\pi,\pi')+\gamma^{-1}v(G(\pi',w(1)),\gamma e^{-\alpha})\right)\right].
\end{equation}
Now, noting that $\cS$ is compact, we get continuity of $(\pi,\gamma)\to Tv(\pi,\gamma)$. This concludes the proof of the $C$-Feller property. 

Now, we show that for $v\equiv 0$, the iterated sequence of operators satisfies the Cauchy condition. For any $n\in\bN$ and $\delta\in\Gamma$, we get
\[
T^{n}0(\pi(1),\delta)=\delta\sup_{\pi}\mu^{\delta}\Bigg(\sum_{t=0}^{n-1}e^{-t\alpha}\left[\ln\langle\pi(t),w(t)\rangle+\ln s(\pi(t-),\pi(t))\right]\Bigg).
\]
For brevity, and with slight abuse of notation, let us introduce the abbreviated notation $Z_{t}(\pi):=\ln\langle\pi(t),w(t)\rangle+\ln s(\pi(t-),\pi(t))$. For any $n,k\in\bN$ and $\delta\in\Gamma$ (i.e. $\gamma<\delta<0$), noting that entropic risk is additive for independent random variables and doing similar calculations as in Lemma~\ref{lm:z.bounds}, we get
\begin{align}
T^{n+k}0(\pi(1),\delta) &=\delta\sup_{\pi}\mu^{\delta}\Bigg(\sum_{t=0}^{n+k-1}e^{-t\alpha}Z_t(\pi)\Bigg)\nonumber\\
& \leq \delta\sup_{\pi}\mu^{\delta}\Bigg(\sum_{t=0}^{n-1}e^{-t\alpha}Z_t(\pi)+\sum_{t=n}^{n+k-1}e^{-t\alpha}\left( \min_i r_i(t)+\ln \tilde s\right)\Bigg)\nonumber\\
& =T^{n}0(\pi(1),\delta)+\delta\sum_{t=n}^{n+k-1}\mu^{\delta}\left(e^{-t\alpha}\left(\min_i r_i(t)+\ln \tilde s\right)\right)\nonumber\\
& \leq T^{n}0(\pi(1),\delta)+\delta\sum_{t=n}^{n+k-1}e^{-t\alpha}\left(\mu^{\gamma}\left(\min_i r_i(t)\right)+\ln \tilde s\right)\nonumber\\
& \leq T^{n}0(\pi(1),\delta)+\gamma z^{-} \frac{e^{-n\alpha}}{1-e^{-\alpha}}.\label{eq:T0n1}
\end{align}
Similarly, we get
\begin{align}
T^{n+k}0(\pi(1),\delta) &=\delta\sup_{\pi}\mu^{\delta}\Bigg(\sum_{t=0}^{n+k-1}e^{-t\alpha}Z_t(\pi)\Bigg)\nonumber\\
& \geq \delta\sup_{\pi}\mu^{\delta}\Bigg(\sum_{t=0}^{n-1}e^{-t\alpha}Z_t(\pi)+\sum_{t=n}^{n+k-1}e^{-t\alpha}\left( \max_i r_i(t)+\ln 1\right)\Bigg)\nonumber\\
& =T^{n}0(\pi(1),\delta)+\delta\sum_{t=n}^{n+k-1}\mu^{\delta}\left(e^{-t\alpha}\max_i r_i(t)\right)\nonumber\\
& \geq T^{n}0(\pi(1),\delta)+\delta\sum_{t=n}^{n+k-1}e^{-t\alpha}\mu^{0}\left(\max_i r_i(t)\right)\nonumber\\
& \geq T^{n}0(\pi(1),\delta)+\gamma z^{+} \frac{e^{-n\alpha}}{1-e^{-\alpha}}.\label{eq:T0n2}
\end{align}
Consequently, combining \eqref{eq:T0n1} and \eqref{eq:T0n2}, for any $n,k\in\bN$ we get
\[
\sup_{\pi\in\cS}\sup_{\delta\in \Gamma}\left|T^{n+k}0(\pi,\delta)-T^{n}0(\pi,\delta)\right|\leq (z^{+}-z^{-}) \frac{|\gamma|, e^{-n\alpha}}{1-e^{-\alpha}},
\]
which shows that the sequence of functions $(T^n0)_{n\in\bN}$ satisfies the Cauchy condition. Now, note that for any $n\in\bN$, function $T^n0$ is continuous and bounded due to Feller property. Consequently, as the space of bounded and continuous functions on $\cS\times \Gamma$ is a Banach space, we know that there exists a bounded and continuous function $v_{\alpha}\colon \cS\times\Gamma\to\bR$ such that
\begin{equation}\label{eq:Tn.limit}
\sup_{\pi\in\cS}\sup_{\delta\in \Gamma}\left|T^{n}0(\pi,\delta)-v_{\alpha}(\pi,\delta)\right|\to 0\quad \textrm{ as } n\to\infty.
\end{equation}
Now, noting that $T^{n+1}0\equiv T(T^n0)$, we get that $T^{n+1}0(\pi,\delta)\to v_{\alpha}(\pi,\gamma)$ as well as  $T^{n+1}0(\pi,\delta)\to Tv_{\alpha}(\pi,\gamma)$, which shows that $v_{\alpha}$ is a fixed point of operator $T$. 
\end{proof}

\begin{remark}
Note that in the proof of Theorem~\ref{th:discounted} we have in fact showed a more direct formula for $v_{\alpha}$ as due to \eqref{eq:Tn.limit} we know that there exists a limit of $T^n0$, as $n\to\infty$, and we have
\[
v_{\alpha}(\pi,\gamma)=\lim_{n\to\infty}\left[\gamma\sup_{\pi}\mu^{\gamma}\Bigg(\sum_{t=0}^{n}e^{-t\alpha}\left[\ln\langle\pi(t),w(t)\rangle+\ln s(\pi(t-),\pi(t))\right]\Bigg)\right].
\]
\end{remark}

\noindent In the end of this section let us show a supplementary result linked to $v_{\alpha}$.
\begin{lemma}\label{lem:v.alpha.span}
Let $\gamma\in \bR\setminus\{0\}$, $\alpha\in (0,1)$, and let $v_{\alpha}$ solve \eqref{eq:Bellman1.d}. Then, for any $\pi,\bar\pi \in\cS$ and $\delta\in\Gamma$ we have
\begin{equation}\label{eq:v.alpha.bounds}
\inf_{\pi'}\delta\ln\left({s(\pi,\pi') \over s(\bar{\pi},\pi')}\right)\leq v_\alpha(\pi,\delta)-v_\alpha(\bar{\pi},\delta)\leq \sup_{\pi'}\delta \ln\left({s(\pi,\pi') \over s(\bar{\pi},\pi')}\right).
\end{equation}
\end{lemma}
\begin{proof}
The proof of \eqref{eq:v.alpha.bounds} follows directly from \eqref{eq:Bellman1.d}. Indeed, using \eqref{eq:Bellman1.d} it is easy to note that for any $\pi\in\cS$, we get
\begin{align*}
v_{\alpha}(\pi,\delta)& \leq \delta\sup_{\pi'\in \mathcal S}\left[ \mu^{\delta} \Big( \ln\langle \pi', w(1)\rangle+\delta^{-1}v_{\alpha}(G(\pi',w(1)),\delta e^{-\alpha})\Big)\right]+\sup_{\pi'\in\cS}\delta\left[\ln s(\pi,\pi')\right],\\
v_{\alpha}(\pi,\delta)& \geq \delta\sup_{\pi'\in \mathcal S}\left[ \mu^{\delta} \Big( \ln\langle \pi', w(1)\rangle+\delta^{-1}v_{\alpha}(G(\pi',w(1)),\delta e^{-\alpha})\Big)\right]+\inf_{\pi'\in\cS}\delta\left[\ln s(\pi,\pi')\right].
\end{align*}
from which \eqref{eq:v.alpha.bounds} follows.
\end{proof}

\section{Vanishing discount approach}\label{S:vanishing}
Fix $\hat{\pi}\in {\cal S}$ and for any $\alpha\in(0,1)$ and $n\in\bN$ define
\begin{align}
\bar{v}_\alpha(\pi,\gamma) & :=v_\alpha(\pi,\gamma)-v_\alpha(\hat{\pi},\gamma),\label{eq:bar.v.alpha}\\
\lambda^{(n)}_{\alpha}& :=v_\alpha(\hat{\pi},\gamma e^{-\alpha n})-v_\alpha(\hat{\pi},\gamma e^{-\alpha (n+1)}).\label{eq:lambda.n.alpha}\\
v^{(n)}_{\alpha}(\pi,\gamma) & :=\bar v_{\alpha}(\pi,\gamma e^{-\alpha n}).
\end{align}
where $v_{\alpha}$ is a solution to the discounted Bellman equation \eqref{eq:Bellman1.d}. First, let us show that sequence introduced in \eqref{eq:lambda.n.alpha} is uniformly bounded.

\begin{lemma}\label{lm:lambda.bound}
Let $\gamma\in\bR\setminus\{0\}$ and let us assume \eqref{A.1}. Then,
\[
\sup_{\alpha\in (0,1)}\sup_{n\in\bN}|\lambda_{\alpha}^{(n)}| <\infty.
\]
\end{lemma}
\begin{proof}
We only show the proof for $\gamma<0$; the proof for $\gamma>0$ is analogous. Let us fix $n\in\bN$ and $\alpha\in (0,1)$. Using Lemma~\ref{lem:v.alpha.span}, for any $n\in\bN$ and  $\pi'\in S$, we get
\begin{equation}\label{eq:Bellman1.ineq2}
|v_{\alpha}(G(\pi',w(1)),\gamma e^{-\alpha (n+1)})-v_{\alpha}(\hat\pi,\gamma e^{-\alpha (n+1)})|  \leq |\gamma|e^{-\alpha n} \cdot |\ln\tilde s|, 
\end{equation}
Consequently, recalling that $v_{\alpha}$ is a solution to the discounted Bellman equation, rewriting $v_\alpha(\hat{\pi},\gamma e^{-\alpha n})$ using \eqref{eq:Bellman1.d}, applying \eqref{eq:Bellman1.ineq2}, and then Lemma~\ref{lm:z.bounds}, we get
\begin{align*}
\lambda_{\alpha}^{(n)} &=v_\alpha(\hat{\pi},\gamma e^{-\alpha n})-v_\alpha(\hat{\pi},\gamma e^{-\alpha (n+1)}) \\
&\leq \gamma e^{-\alpha n}\sup_{\pi'\in \mathcal S}\left[ \mu^{\gamma e^{-\alpha n}} \Big( \ln\langle \pi', w(1)\rangle+\ln s(\hat\pi,\pi')-(\gamma e^{-\alpha n})^{-1}(\gamma e^{-\alpha n}) \cdot |\ln\tilde s|\Big)\right]\\
& \leq \gamma e^{-\alpha n}\left(\sup_{\pi'\in \mathcal S} z_{\gamma e^{-\alpha n}}(\hat\pi,\pi')-|\ln\tilde s|\right)\\
& \leq \gamma \left(z^{-}-|\ln\tilde s|\right).
\end{align*}
Similarly, we get $\lambda_{\alpha}^{(n)} \geq \gamma \left(|\ln\tilde s|+z^{+}\right)$. Noting that both upper and lower bound is independent of $\alpha$ and $n$, we conclude the proof.
\end{proof}
Now, we present two main results of this section, which shows that under \eqref{A.1} one could find a sequence of functions solving the iterated Bellman equation. These functions could be used to find optimal strategy and related optimal value for the problem \eqref{eq:obj1}. While Theorem~\ref{th:main1} is in fact true under both risk-averse ($\gamma<0$) and risk-seeking ($\gamma>0$) case, in the latter case we can show that iteration is not required, i.e. we can directly solve~\eqref{eq:Bellman0.2}; this is stated in Theorem~\ref{th:main2}.

\begin{theorem}\label{th:main1}
Let $\gamma<0$ and let us assume \eqref{A.1}. Then, there exists a sequence of constants $\lambda^{(n)}$, $n\in\bN$, and a sequence of continuous bounded functions $v^{n}(\cdot,\gamma)\colon\cS\to \bR$, $n\in\bN$, such that the recursive Bellman equation
\[
v^{(n)}(\pi,\gamma) +\lambda^{(n)}=\gamma \sup_{\pi'\in \mathcal S}\bigg[ \mu^{\gamma} \Big( \ln\langle \pi', w(1)\rangle+\ln s(\pi,\pi')+\gamma ^{-1}v^{(n+1)}(G(\pi',w(1)),\gamma)\Big)\bigg]
\]
is satisfied for any $n\in\bN$. Moreover, the constant $\Lambda:=\liminf_{n\to\infty}\lambda^{(n)}/ (\gamma n)$ is the optimal value for the problem \eqref{eq:obj1}, i.e. we get $\Lambda=\sup_{\pi}J(\pi)$, and the optimal (iterated) strategy is defined by the selectors to the recursive Bellman equation. \end{theorem}
\begin{proof}
Fix $\gamma<0$. First, observe that the family of functions $\{\pi\to \bar v_{\alpha}(\pi,\delta)\}$, indexed by $\alpha\in (0,1)$ and $\delta\in\Gamma$ is both uniformly bounded and equicontinuous. Indeed, both uniform boundedness and equicontinuity follows directly from Lemma~\ref{lem:v.alpha.span} since for any $\pi,\tilde\pi \in \cS$ we get
\begin{align*}
\sup_{\alpha\in (0,1)}\sup_{\delta\in \Gamma} |\bar v_{\alpha}(\pi,\delta)|& \leq |\gamma|\cdot \left|\ln \tilde s\right|,\\
\sup_{\alpha\in (0,1)}\sup_{\delta\in \Gamma}|\bar v_{\alpha}(\pi,\delta)-\bar v_{\alpha}(\tilde \pi,\delta)| & <|\gamma| \sup_{\pi'\in\cS}\left|\ln \frac{s(\pi,\pi')}{s(\tilde\pi,\pi')}\right|.
\end{align*}
Thus, by the Arzela-Ascoli theorem, we know that there exists a decreasing sequence $(\alpha_i)_{i\in\bN}$, such that $\alpha_i\in (0,1)$, $\alpha_i \searrow 0$ as $i\to\infty$, and for any $n\in\bN$, we have
\begin{equation}\label{eq:arzela1}
v^{(n)}_{\alpha}(\pi,\gamma) \to v^{(n)}(\pi,\gamma),\quad \textrm{$\pi$-uniformly}, 
\end{equation}
for some continuous and bounded function $v^{(n)}(\cdot,\gamma)\colon \cS\to\bR$. Second, from Lemma~\ref{lm:lambda.bound}, we know that the sequence $(\alpha_i)_{i\in\bN}$, can be chosen in such a way that for any $n\in\bN$ we also have 
\begin{equation}\label{eq:arzel2}
\lambda^{(n)}_{\alpha_i} \to \lambda^{(n)},\quad i\to\infty,
\end{equation}
where $(\lambda^{(n)})_{n\in\bN}$ is some sequence of real numbers. Third, by combining  Bellman equation \eqref{eq:Bellman1.d} with \eqref{eq:bar.v.alpha} and \eqref{eq:lambda.n.alpha}, for any $n\in\bN$, we get
\begin{multline}\label{eq:limit.n.i}
v^{(n)}_{\alpha_i}(\pi,\gamma) +\lambda_{\alpha_i}^{(n)}=\gamma e^{-n\alpha_i}\sup_{\pi'\in \mathcal S}\bigg[ \mu^{\gamma e^{-n\alpha_i}} \Big( \ln\langle \pi', w(1)\rangle+\ln s(\pi,\pi')\\
+(\gamma e^{-n\alpha_i})^{-1}v^{(n+1)}_{\alpha_i}(G(\pi',w(1)),\gamma)\Big)\bigg].
\end{multline}
Now, noting that the limit values are also bounded, for each $n\in\bN$, we can take the limit in \eqref{eq:limit.n.i}, as $i\to\infty$, and get
\begin{equation}\label{eq:iteration}
v^{(n)}(\pi,\gamma) +\lambda^{(n)}=\gamma \sup_{\pi'\in \mathcal S}\bigg[ \mu^{\gamma} \Big( \ln\langle \pi', w(1)\rangle+\ln s(\pi,\pi')+\gamma ^{-1}v^{(n+1)}(G(\pi',w(1)),\gamma)\Big)\bigg],
\end{equation}
which concludes the first part of the proof. Now, iterating the sequence \eqref{eq:iteration}, starting from $n=1$, and using tower property of entropic utility, for any $n\in\bN$, we get
\[
\sum_{i=1}^n \lambda^{(n)}=\gamma \sup_{\pi}\bigg[\mu^{\gamma} \Big( \sum_{t=1}^{n-1} Z_{t}(\pi)+\gamma ^{-1}\left[v^{(n+1)}(G(\pi(n),w(n)),\gamma)-v^{(0)}(\pi(0),\gamma)\right]\Big)\bigg]
\]
where $Z_{t}(\pi):=\ln\langle\pi(t),w(t)\rangle+\ln s(\pi(t-),\pi(t))$. Dividing both sides by $\frac{1}{\gamma n}$, noting that the sequence of functions $(v^{(n)}$ is uniformly bounded by $\pm |\gamma|\cdot |\tilde s|$, and taking the limes inferior of both sides, we get
\[
\Lambda=\sup_{\pi}\liminf_{n\to\infty}\frac{1}{n}\mu^{\gamma} \left( \sum_{t=1}^{n-1} Z_{t}(\pi)\right)=\sup_{\pi}J(\pi),
\]
which concludes the proof. Also, note that for any admissible strategy $\pi$ we get $\Lambda\geq J(\pi)$, while for the strategy $\tilde \pi$ determined by the iterated sequence $(v^{(n)})$ we get $\Lambda= J(\tilde \pi)$.
\end{proof}

\noindent Now, let us show that for $\gamma>0$, the result in Theorem~\ref{th:main1} could be strengthened in a sense that recursive scheme is not required and one can solve directly \eqref{eq:Bellman0.2}.

\begin{theorem}\label{th:main2}
Let $\gamma>0$ and let us assume \eqref{A.1}. Then, there exists a constant $\lambda$ and a continuous bounded functions $v(\cdot,\gamma)\colon\cS\to \bR$ that solves Bellman's equation \eqref{eq:Bellman0.2}, i.e. we get
\[
v(\pi,\gamma) +\lambda=\gamma \sup_{\pi'\in \mathcal S}\bigg[ \mu^{\gamma} \Big( \ln\langle \pi', w(1)\rangle+\ln s(\pi,\pi')+\gamma ^{-1}v(G(\pi',w(1)),\gamma)\Big)\bigg].
\]
Moreover, the constant $\Lambda:=\lambda/ \gamma$
is the optimal value for the problem \eqref{eq:obj1}, i.e. we get $\Lambda=\sup_{\pi}J(\pi)$, and the optimal strategy is defined by the selectors to the Bellman equation.
\end{theorem}

\begin{proof}
Fix $\gamma>0$. The first part of the proof is analogous to the proof of Theorem~\ref{th:main2}. Applying similar reasoning, we get that there exists sequence of constants $\lambda^{(n)}$ and bounded continuous functions $v^{(n)}(\cdot,\gamma)\colon  \cS\to\bR$, $n\in\bN$ such that 
\begin{equation}\label{eq:iteration.geq}
v^{(n)}(\pi,\gamma) +\lambda^{(n)}=\gamma \sup_{\pi'\in \mathcal S}\bigg[ \mu^{\gamma} \Big( \ln\langle \pi', w(1)\rangle+\ln s(\pi,\pi')+\gamma ^{-1}v^{(n+1)}(G(\pi',w(1)),\gamma)\Big)\bigg],
\end{equation}
Now, let us show that for $\gamma>0$ the sequence $(\lambda^{(n)})_{n\in\bN}$ is non-decreasing. First, note that for any random variable $Z$, the mapping $\gamma\to \ln\bE[e^{\gamma Z}]$ is convex.\footnote{This could be easily shown using H{\"o}lder inequality by considering $U=e^{(1-\theta)\gamma Z}$ and $V=e^{\theta\gamma Z}$, for $\theta\in [0,1]$, with $\frac{1}{p}+\frac{1}{q}=1$, such that $p=\frac{1}{1-\theta}$ and $q=\frac{1}{\theta}$, and taking logarithm of both sides.} Consequently, since supremum of a family of convex functions is convex, we get that for any $\alpha\in (0,1)$ and $\pi\in\cS$ the mapping $\gamma\to v_{\alpha}(\pi,\gamma)$ is also convex. In particular, for any $n\in\bN$, we get
\begin{equation}\label{eq:convex1}
\frac{v_{\alpha}(\pi,\gamma e^{-\alpha n})-v_{\alpha}(\pi,\gamma e^{-\alpha (n+1)})}{\gamma e^{-\alpha n}(1-e^{-\alpha})}\geq \frac{v_{\alpha}(\pi,\gamma e^{-\alpha (n+1)})-v_{\alpha}(\pi,\gamma e^{-\alpha (n+2)})}{\gamma e^{-\alpha (n+1)}(1-e^{-\alpha})},
\end{equation}
which implies $e^{-\alpha}\lambda_{\alpha}^{(n)}\leq \lambda_{\alpha}^{(n+1)}$. Now, letting $i\to\infty$ in the decreasing sequence $(\alpha_i)_{i\in\bN}$ defined in the proof of Theorem~\ref{th:main1}, we get
\[
\lambda^{(n)}=\lim_{i\to\infty}e^{-\alpha_i}\lambda_{\alpha_i}^{(n)} \leq \lim_{i\to\infty}\lambda_{\alpha_i}^{(n+1)}=\lambda^{(n+1)},
\]
which concludes this part of the proof. Second, from \eqref{eq:convex1} we get
\[
e^{-\alpha}(v^{(n+1)}_{\alpha}(\pi,\gamma)-v^{(n)}_{\alpha}(\pi,\gamma) +\lambda^{(n)}_{\alpha}) \leq v^{(n+2)}_{\alpha}(\pi,\gamma)-v^{(n+1)}_{\alpha}(\pi,\gamma )+\lambda^{(n+1)}_{\alpha}.
\]
Again, taking the limit $i\to\infty$, for the decreasing sequence $(\alpha_i)_{i\in\bN}$ defined in the proof of Theorem~\ref{th:main1}, we get
\[
v^{(n+1)}(\pi,\gamma)-v^{(n)}(\pi,\gamma) +\lambda^{(n)} \leq v^{(n+2)}(\pi,\gamma)-v^{(n+1)}(\pi,\gamma )+\lambda^{(n+1)}.
\]
Consequently, the sequence of functions $z^{(n)}(\pi,\gamma):=v^{(n+1)}(\pi,\gamma)-v^{(n)}(\pi,\gamma) +\lambda^{(n)}$ is increasing wrt. $n$. As the sequence $z^{(n)}$ is  equicontinuous and bounded, there exists a continuous bounded function $z(\cdot,\gamma)\colon \cS\to\bR$ such that $z^{(n)}(\cdot,\gamma)\to z(\cdot,\gamma)$, as $n\to\infty$. Now, since $\lambda^{(n)} \nearrow \lambda$ for some $\lambda\in \bR$, as $n\to\infty$, we get
\[
\left[v^{(n+1)}(\pi,\gamma)-v^{(n)}(\pi,\gamma)\right] \to z(\pi,\gamma)-\lambda, \quad n\to\infty.
\]
Now, note that for any $\pi\in\cS$ we get $z(\pi,\gamma)=\lambda$ as otherwise, the sequence $(v^{(n)}(\pi,\gamma))_{n\in\bN}$, that could be represented by
\[
v^{(n+1)}(\pi,\gamma)=\sum_{i=1}^{n} \left[v^{(i+1)}(\pi,\gamma)-v^{(i)}(\pi,\gamma)\right]+v^{(1)}(\pi,\gamma), \quad n\in\bN,
\]
would be unbounded which would lead to contradiction as $|v^{(n)}(\pi,\gamma)|<|\gamma|\cdot |\ln\tilde s|$, for $n\in\bN$. This implies
\begin{equation}\label{eq:limit.zero1}
\lim_{n\to\infty} \left[v^{(n+1)}(\pi,\gamma)-v^{(n)}(\pi,\gamma)\right] \to 0,\quad n\to\infty.
\end{equation}
Due to Arzela-Ascoli theorem, as the mapping $n\to v^{(n)}(\cdot,\gamma)$ is equicontinuous and uniformly bounded, we can choose a subsequence $(n_k)_{k\in\bN}$ and continuous bounded function $v(\cdot,\gamma)\colon \cS\to\bR$ such that
\[
v^{(n_k)}(\pi,\gamma)\to v(\pi,\gamma),\quad \textrm{$\pi$-uniformly}.
\]
Finally, as $\lambda^{(n)} \nearrow \lambda$, $n\to\infty$, recalling \eqref{eq:limit.zero1}, and letting $n\to\infty$ in \eqref{eq:iteration} we get 
\[
v(\pi,\gamma) +\lambda=\gamma \sup_{\pi'\in \mathcal S}\bigg[ \mu^{\gamma} \Big( \ln\langle \pi', w(1)\rangle+\ln s(\pi,\pi')+\gamma ^{-1}v(G(\pi',w(1)),\gamma)\Big)\bigg].
\]
which concludes the proof.
\end{proof}

\section{Span-contraction approach}\label{S:span-contraction}
In Section~\ref{S:vanishing} we have shown that for $\gamma>0$ one can solve directly Bellman equation \eqref{eq:Bellman0.2}; see Theorem~\ref{th:main2}. On the other hand, for $\gamma<0$, we were only able to obtain the recursive scheme as presented in Theorem~\ref{th:main2}. In this section, we show that the solution to \eqref{eq:Bellman0.2} exists also for $\gamma<0$ under relatively weak ergodic assumptions imposed on asset log-returns. We follow the span-contraction approach; see e.g.~\cite{PitSte2016}. In the span-contraction approach the parameter $\gamma<0$ is kept fixed in a sense that we do not need to introduce the discounting scheme. Consequently, to ease the exposition, rather than using notation from \eqref{eq:Bellman0.2}, we revert to the one from \eqref{eq:Bellman0}: we fix one $\gamma<0$ and write $v(\pi)$ rather than $v(\pi,\gamma)$. As usual, we use $B(\cS)$ to denote the space of continuous and bounded functions $v\colon \cS \to\bR$. For any $v\in B(\cS)$ we introduce supremum norm and span semi-norm notation
\[
\|v\|:=\sup_{\pi\in\cS}|v(\pi)|\quad\textrm{and}\quad \|v\|_{\textrm{sp}}:=\sup_{\pi,\pi'\in\cS}\frac{v(\pi)-v(\pi')}{2}.
\]
Note that those norms are bound by relation
\begin{equation}\label{eq:constant}
\inf_{d\in\bR}\|v+d\|=\|v\|_{\textrm{sp}},
\end{equation}
see \cite{PitSte2016} or \cite{HaiMat2011} for details. Now, we introduce additional assumption that relates to ergodicity and plays a central role in the span-contraction approach. Namely, for any $\delta\in (0,1/d)$, we assume that
\begin{equation}\label{A.2} \tag{A.2}
\sup_{A\in \mathcal{B}(S)}\sup_{\pi,\pi'\in \cS_{\delta}}\Big( \bP[G(\pi,w(1))\in A]-\bP[G(\pi',w(1))\in A] \Big)<1,
\end{equation}
where $\cS_{\delta}:=\{\pi\in S\colon \min_{i} \pi_i\geq \delta\}$ identifies a set of strategies in which we allocate at least $\delta$ proportion of capital to each asset. This assumption is related to mixing and states that whatever our initial (non-degenerated) allocation is, we expect to be in some common set with positive probability.

\begin{remark}[Ergodicity/mixing assumption relevance]
Recalling that $G(x,y)=x\cdot y / \langle x,y\rangle$, for any $\pi\in S$ we get $G(\pi,w(1))=\frac{\pi \cdot w(1)}{\langle \pi ,w(1)\rangle}$ which shows that \eqref{A.2} is in fact related to assumptions imposed on $w(1)$. Nevertheless, we decided to present \eqref{A.2} in its classical form, to show the connection to mixing. One can show that assumption \eqref{A.2} is satisfied by any log Levy process, even with ergodic economic factors, see \cite[Proposition 1]{DunDunSte2011} for details. Also, one can notice that if $r(1)$ has full support, then the assumption \eqref{A.2} is automatically satisfied.
\end{remark}
Now, for any $\delta\in (0,1/d)$ we introduce operator
\[
T_{\delta}v(\pi):=\sup_{\pi'\in \mathcal \cS_{\delta}}\left[\gamma \mu^{\gamma} \Big( \ln\langle \pi', w(1)\rangle+\ln s(\pi,\pi')+\gamma^{-1}v(G(\pi',w(1)))\Big)\right],\quad v\in B(\cS).
\]
It is relatively easy to show that operator $T_{\delta}$ is $C$-Feller. In particular, using similar reasoning as in Lemma~\ref{lem:v.alpha.span}, for any $v\in B(\cS)$ we get

\begin{equation}\label{eq:span.contr.bounded}
\|T_{\delta}v\|_{\textrm{sp}}\leq |\gamma| \sup_{\pi,\pi',\tilde\pi\in \cS}\ln\left[\ \frac{s(\pi,\pi')}{s(\tilde\pi,\pi')} \right]\leq -|\gamma| \ln \tilde s:=K,
\end{equation}
which implies boundedness of $T_{\delta}v$, for any $v\in B(\cS)$. Let us now show that $T_{\delta}$ is a local contraction.

\begin{proposition}
Let $\gamma<0$ and let us assume \eqref{A.1} and \eqref{A.2}. Then, for each $\delta\in (0,1\d)$, the operator $T_{\delta}$ is a local contraction under $\|\cdot\|_{\textrm{sp}}$, i.e. there exists $L_{\delta}:\bR_{+}\to (0,1)$ such that
\[
\|T_{\delta}v_1 -T_{\delta}v_2\|_{\textrm{sp}} \leq L_{\delta}(M)\|v_1-v_2\|_{\textrm{sp}},
\]
for $v_1,v_2\in C(\cS)$, such that $\|v_1\|\leq M$ and $\|v_2\|\leq M$.
\end{proposition}

\begin{proof}
Let us fix $\delta\in (0,1/d)$. For any $v\in B(\cS)$ and $\pi\in \cS$ let
\[
\bar\mu_{(\pi,v)}(B):=\frac{\bE\left[1_B(G(\pi,w(1))e^{\gamma\ln\langle \pi, w(1)\rangle+v(G(\pi,w(1)))} \right]}{\bE\left[e^{\gamma\ln\langle \pi, w(1)\rangle+v(G(\pi,w(1)))} \right]},\quad B\in \cB(\bR^d)
\]
denote the projection measure for $v$ with rebalancing $\pi$, and let
\[
\pi_{v}:=\arg\max_{\pi\in \cS_{\delta}}\left[\gamma \mu^{\gamma} \Big( \ln\langle \pi', w(1)\rangle+\ln s(\pi,\pi')+\gamma^{-1}v(G(\pi',w(1)))\Big)\right]
\]
denote the optimal rebalancing (induced by operator $T_{\delta}$) given $v\in C(S)$ and initial state $\pi \in \cS$. First, let us show that for any $v_1,v_2\in C(\cS)$ and $\pi,\pi'\in \cS_{\delta}$ we get
\begin{equation}\label{eq:loc.contr0}
T_{\delta}v_1(\pi)-T_{\delta}v_2(\pi)-(T_{\delta}v_1(\pi')-T_{\delta}v_2(\pi'))\leq \|v_1-v_2\|_{\textrm{sp}}\cdot \|\bH_{\pi,\pi'}^{v_1,v_2}\|_{\textrm{var}},
\end{equation}
where $\bH_{\pi,\pi'}^{v_1,v_2}:=\bar\mu_{(\pi_{v_2},v_1)}-\bar\mu_{(\pi'_{v_1},v_2)}$ is a signed measure and $\|\cdot\|_{\textrm{var}}$ is the total variation norm given by
\[
\|\bH \|_{\textrm{var}}:=\int_{\bR^d}|\bH|(d x),
\]
for $|\bH|$ being the total variation of $\bH$; see \cite{PitSte2016} for details. Using dual representation for entropic risk, and performing similar calculations as in  \cite[Lemma 1]{PitSte2016} we get 
\[
T_{\delta}v_1(\pi)-T_{\delta}v_2(\pi)-(T_{\delta}v_1(\pi')-T_{\delta}v_2(\pi'))\leq \int_{\bR^d}\left[v_1(x)-v_2(x) \right]\, \bH_{\pi,\pi'}^{v_1,v_2}(d x).
\]
Now, recalling that for any $v\in B(\cS)$ we have $\inf_{d\in\bR}\|v+d\|=\|v\|_{\textrm{sp}}$, we know there exists $d\in\bR$ such that
\[
 \|v_1-v_2\|_{\textrm{sp}}=\sup_{x\in\bR^d}(v_1(x)-v_2(x)+d)=-\inf_{x\in\bR^d}(v_1(x)-v_2(x)+d)
\]
Consequently, for $A$ denoting the positive set of measure $\bH_{\pi,\pi'}^{v_1,v_2}$, we get
\begin{align*}
\int_{\bR^d}\left[v_1(x)-v_2(x) \right]\, \bH_{\pi,\pi'}^{v_1,v_2}(d x) & = \int_{\bR^d}\left[v_1(x)-v_2(x) +d\right]\, \bH_{\pi,\pi'}^{v_1,v_2}(d x)\\
&\leq \|v_1-v_2\|_{\textrm{sp}}\left(\int_{A} \bH_{\pi,\pi'}^{v_1,v_2}(d x)-\int_{A^c} \bH_{\pi,\pi'}^{v_1,v_2}(d x)\right)\\
& = \|v_1-v_2\|_{\textrm{sp}}\cdot \|\bH_{\pi,\pi'}^{v_1,v_2}\|_{\textrm{var}},
\end{align*}
which concludes this step of the proof. Now, let us fix $M\in\bR_{+}$ and show that there exists constant $L(M)\in (0,1)$ such that for any $v_1,v_2\in C(\cS)$ satisfying $\|v_1\|\leq M$ and $\|v_2\|\leq M$, and $\pi,\pi'\in \cS$, we get
\begin{equation}\label{eq:loc.contr1}
\|\bH_{\pi,\pi'}^{v_1,v_2}\|_{\textrm{var}} \leq 2 L(M).
\end{equation}
Suppose that \eqref{eq:loc.contr1} is not satisfied. Recalling that
\[
\|\bH_{\pi,\pi'}^{v_1,v_2}\|_{\textrm{var}} =2\sup_{B\in \cB(\bR^d)}|\bar\mu_{(\pi_{v_2},v_1)}(B)-\bar\mu_{(\pi'_{v_1},v_2)}(B)|,
\]
we get that there exists a sequence of sets $B_n\in \cB(\bR^d)$, functions $v_n,v_n'\in B(\cS)$ satisfying $\|v_n\|\leq M$ and $\|v'_n\|\leq M$, and weights $\pi_n,\pi_n' \in \cS$ such that
\begin{equation}\label{eq:loc.contr2}
\bar\mu_{(\pi_{v'_n},v_n)}(B_n)\to 1\quad \textrm{and}\quad \bar\mu_{(\pi_{v_n},v'_n)}(B_n)\to 0,\quad \textrm{as } n\to\infty. 
\end{equation}
Now, let $r_+(1):=\max_i r_i(1)$ and $r_-(1):=\min_i r_i(1)$. For any $B\in \cB(\bR^d)$, $\pi\in \cS$ and $v\in B(\cS)$, we get
\begin{align}
\bar\mu_{(\pi,v)}(B) &= \frac{\bE\left[1_B(G(\pi,w(1))e^{\gamma\ln\langle \pi, w(1)\rangle+v(G(\pi,w(1)))} \right]}{\bE\left[e^{\gamma\ln\langle \pi, w(1)\rangle+v(G(\pi,w(1)))} \right]}\nonumber\\
&\geq \frac{\bE\left[1_B(G(\pi,w(1))e^{\gamma\ln\langle \pi, w(1)\rangle+\inf_{\pi'\in\cS}v(\pi)} \right]}{\bE\left[e^{\gamma\ln\langle \pi, w(1)\rangle+\sup_{\pi'\in\cS}v(\pi')} \right]}\nonumber\\
&= e^{-2\|v\|_{\textrm{sp}}}\frac{\bE\left[1_B(G(\pi,w(1)))e^{\gamma\ln\langle \pi, w(1)\rangle} \right]}{\bE\left[e^{\gamma\ln\langle \pi, w(1)\rangle} \right]}\nonumber\\
&\geq e^{-2\|v\|_{\textrm{sp}}}\frac{\bE\left[1_{\{G(\pi,w(1))\in B\}}e^{\gamma r_+(1)} \right]}{\bE\left[e^{\gamma r_-(1)} \right]}.\label{eq:loc.contr3}
\end{align}
Let us now assume there exists $\epsilon\in (0,1)$, such that $\bE[1_B(G(\pi,w(1)))]>\epsilon$ and let $R(\epsilon'):=F^{[-1]}_{r_+(1)}(1-\epsilon')$, where $F^{[-1]}_{r_+(1)}$ is the generalised (lower) quantile function of $r_+$. First, if there exists $\epsilon' \in (0,1)$ such that $\epsilon'\leq \epsilon$ and $\bP[\{r_+(1)\geq  R(\epsilon')\}]\leq \epsilon $,  then we have 
\begin{align}
\bE\left[1_{\{G(\pi,w(1))\in B\}}e^{\gamma r_+(1)}\right] &\geq  \bE\left[1_{\{r_+(1)\geq R(\epsilon')\}}e^{\gamma r_+(1)} \right]\nonumber\\
&\geq  \bE\left[1_{\{R(\epsilon'/2)\geq r_+(1)\geq R(\epsilon')\}}e^{\gamma r_+(1)} \right]\nonumber\\
&\geq \epsilon'/2 \cdot e^{\gamma R(\epsilon'/2)}.\label{eq:loc.contr3b}
\end{align}
Second, let us assume there is no $\epsilon'\in (0,1)$ such that  $\epsilon'\leq \epsilon$ and $\bP[\{r_+(1)\geq  R(\epsilon')\}]\leq \epsilon $. Then, we have $\bP[\{r_+(1)\geq  R(\epsilon)\}]= \bP[\{r_+(1)=  R(\epsilon)\}]$ and consequently, for $\epsilon'=2\epsilon$, we get
\begin{equation}\label{eq:loc.contr3bb}
\bE\left[1_{\{G(\pi,w(1))\in B\}}e^{\gamma r_+(1)}\right]\geq e^{\gamma R(\epsilon)}\bP\left[G(\pi,w(1))\in B\right]\geq \epsilon'/2 \cdot e^{\gamma R(\epsilon'/2)}.
\end{equation}
Combining \eqref{eq:loc.contr3}, \eqref{eq:loc.contr3b} and \eqref{eq:loc.contr3bb}, for any $B\in \cB(\bR^d)$, $\pi\in \cS$, $v\in B(\cS)$, and $\epsilon>0$, such that $\bE[1_B(G(\pi,w(1)))]>\epsilon$, there exists $\epsilon'\in (0,1)$ such that 
\begin{equation}\label{eq:loc.contr3c}
\bar\mu_{(\pi,v)}(B) \geq e^{-2\|v\|_{\textrm{sp}}}\frac{\epsilon'/2 \cdot e^{\gamma R(\epsilon'/2)}}{e^{\gamma \cdot \left[(1/\gamma)\ln \bE\left[e^{\gamma r_-(1)} \right]\right]}} \geq e^{-2\|v\|_{\textrm{sp}}} \frac{\epsilon'/2 \cdot  e^{\gamma R(\epsilon'/2)}}{e^{\gamma \mu^{\gamma}(r_-(1))}}.
\end{equation}
From \eqref{A.1}, both $\mu^{\gamma}(r_-(1))$ and $R(\epsilon'/2)$ are finite for any $\epsilon>0$. Also, note that $R$ is increasing and the choice of $\epsilon'$ in \eqref{eq:loc.contr3b}  depends only on the choice of $\epsilon$, i.e. given condition $\bE[1_B(G(\pi,w(1)))]>\epsilon$, the choice of $\epsilon'$ is independent of the choice of $B\in \cB(\bR^d)$, $\pi\in \cS$ and $v\in B(\cS)$. Consequently, combining \eqref{eq:loc.contr2} with \eqref{eq:loc.contr3c}, recalling that $\|v_n\|\leq M$ and $\|v'_n\|\leq M$, for $n\in\bN$, we get  $\bE\left[1_{B'_n}(G(\pi_{v'_n},w(1))) \right]\to 0$ and $\bE\left[1_{B_n}(G(\pi_{v_n},w(1))) \right]\to 0$, as $n\to\infty$. Therefore, as $n\to\infty$, we have
\begin{equation}\label{eq:loc.contr4}
\bP[G(\pi_{v'_n},w(1))\in B_n] - \bP[G(\pi_{v_n},w(1))\in B_n]\to 1.
\end{equation}
Noting that $\pi_{v_n},\pi_{v'_n} \in \cS_{\delta}$, we get that  \eqref{eq:loc.contr4} contradicts \eqref{A.2}, which concludes the proof of \eqref{eq:loc.contr1}. Combining \eqref{eq:loc.contr0} with \eqref{eq:loc.contr1} we conclude the proof.
\end{proof}
Finally, we are ready to show the main result of this section; note that while in Theorem~\ref{th:main3} we adjusted notation from $v(\cdot,\gamma)$ to $v(\cdot)$, the presented conclusions are consistent with those presented in Theorem~\ref{th:main2} (for $\gamma<0$ instead of $\gamma>0$).

\begin{theorem}\label{th:main3}
Let $\gamma<0$ and let us assume \eqref{A.1} and \eqref{A.2}. Then, there exists a constant $\lambda$ and a continuous bounded functions $v(\cdot,\gamma)\colon\cS\to \bR$ that solves \eqref{eq:Bellman0.2}, i.e. we get
\[
v(\pi) +\lambda=\gamma \sup_{\pi'\in \mathcal S}\bigg[ \mu^{\gamma} \Big( \ln\langle \pi', w(1)\rangle+\ln s(\pi,\pi')+\gamma ^{-1}v(G(\pi',w(1)))\Big)\bigg].
\]
Moreover, the constant $\Lambda:=\lambda/ \gamma$
is the optimal value for the problem \eqref{eq:obj1}, i.e. we get $\Lambda=\sup_{\pi}J(\pi)$, and the optimal strategy is defined by the selectors to the Bellman equation.
\end{theorem}

\begin{proof}
For any fixed $\delta\in (0,1/d)$, combining \eqref{eq:span.contr.bounded} with Theorem~\ref{th:main3} we get that there exists a unique (up to an additive constant) function $v_{\delta} \in B(S)$ and a constant $\lambda_{\delta}\in\bR$ such that
\begin{equation}\label{eq:final1}
v_{\delta}(\pi) +\lambda_{\delta}=\gamma \sup_{\pi'\in \cS_{\delta}}\bigg[ \mu^{\gamma} \Big( \ln\langle \pi', w(1)\rangle+\ln s(\pi,\pi')+\gamma ^{-1}v_{\delta}(G(\pi',w(1)))\Big)\bigg].
\end{equation}
Now, let $\hat\pi:=(1/d,\ldots,1/d)$ and $\bar v_{\delta}(\pi):=v_{\delta}(\pi)-v_{\delta}(\hat\pi)$ for $\delta\in (0,1/d)$; note that $\bar v_{\delta}$ also solves \eqref{eq:final1} and we get $\| \bar v_{\delta}\|<2K$ due to \eqref{eq:span.contr.bounded} and property $\bar v_{\delta}(\hat\pi)=0$. Thus, the family of functions $\{\bar v_{\delta}\}_{\delta\in (0,1/d)}$ is uniformly bounded and equicontinuous since, by \eqref{eq:final1}, we get
\begin{equation}
 v_{\delta}(\pi) -  v_{\delta}(\bar \pi) \leq |\gamma| \sup_{\pi'\in \cS_{\delta}}\left|\ln{s(\pi,\pi')\over s(\bar \pi,\pi')}\right|.
\end{equation}
Consequently, using Arzela-Ascoli theorem, we know that there exists a function $v\in B(\cS)$ and a subsequence $(\delta_n)_{n\in\bN}$, $\delta_n\searrow 0$, such that $v_{\delta_n}\to v$ (uniformly) and $\lambda_{\delta_n}\to\lambda$, as $n\to\infty$. Thus, due to uniform convergence, we get
\begin{equation}\label{eq:final2}
v(\pi) +\lambda=\gamma \sup_{\pi'\in \cS_0}\bigg[ \mu^{\gamma} \Big( \ln\langle \pi', w(1)\rangle+\ln s(\pi,\pi')+\gamma ^{-1}v(G(\pi',w(1)))\Big)\bigg],
\end{equation}
where $\cS_0:=\bigcup_{\delta\in (0,1/d)} S_{\delta}$. Now, since right hand side of \eqref{eq:final2} is a continuous function of $\pi'$, we also get
\[
v(\pi) +\lambda=\gamma \sup_{\pi'\in \cS}\bigg[ \mu^{\gamma} \Big( \ln\langle \pi', w(1)\rangle+\ln s(\pi,\pi')+\gamma ^{-1}v(G(\pi',w(1)))\Big)\bigg],
\]
which concludes the proof.
\end{proof}

\section{Numerical examples}\label{S:examples}
In this section we illustrate the implications of Theorem~\ref{th:main2} and Theorem~\ref{th:main3}. Namely, for exemplary dynamics, we approximate the optimal strategies and maximise the objective function~\eqref{eq:obj1} using the associated Bellman's equations. For brevity, we study only the risk-averse case $\gamma<0$ and focus on a simple low-dimensional setting in which one can directly approximate optimal policies following the standard policy iteration algorithm under known dynamics, see \cite{Whi1990} for details. It should be noted that the policy iteration scheme could be also applied to a more realistic high-dimensional market setting in which transition kernel might be unknown, change in time, etc. For an overview of more advanced methods based on Reinforced Learning, $Q$-learning, etc., we refer to \cite{FeiYanWan2021,BasBhaBor2008,Bor2010,AraBisPra2021}.
The goal of this section is to illustrate how to numerically approximate the solution \eqref{eq:Bellman0} and recover the optimal policy used to construct optimal trading strategy. Namely, by considering  iterations of operator $T$ given by
\begin{equation}\label{eq:ex.T}
Tv(\pi)=\sup_{\pi'\in \mathcal S}\left[\mu^{\gamma} \Big( \ln\langle \pi', w(1)\rangle+\ln s(\pi,\pi')+v(G(\pi',w(1)))\Big)\right],\quad v\in B(\cS),
\end{equation}
we get a series of functions $T0$, $T^20$, $\ldots$, $T^n0$ that should converge in the span norm, as $n\to\infty$, to the solution of the Bellman's equation \eqref{eq:Bellman0}. By studying the difference of the consequent iterations, we have control over so called {\it regret} that can tell us how far from the optimal policy we are, see \cite{FeiYanWan2021}. Le us now present two illustrative examples that show why detailed look into Bellman's equation could lead to a more optimal trading choice under transaction costs when confronted with plausible alternatives.

First, we introduce a two-dimensional toy example based on finite state-space dynamics. Its aim is to show that even under simplistic setting, one can meaningfully increase the trading performance by solving a Bellman's equation and easily recover optimal barrier-hit strategy with a simple numerical scheme. 

Second, we introduce a three dimensional example based on correlated Gaussian noise with drift. Note that while in this paper we follow an i.i.d. framework under which very generic assumptions \eqref{A.1} and \eqref{A.2} are sufficient to guarantee the existence to Bellman's equation \eqref{eq:Bellman0}, similar reasoning could be applied to a more generic Markovian setting, see e.g. \cite{PitSte2016} where the discrete-time version of the Merton’s inter-temporal capital asset pricing model (CAPM) is considered.

\begin{example}[Toy example]\label{S:example1}
To illustrate how to approximate Bellman equation's solution and how the solution is linked to portfolio performance let us present a simple two-dimensional example based on finite state-space dynamics. Let $d=2$, $\gamma=-0.5$, and let  $S(t)=S(0)\cdot\prod_{i=1}^{t} e^{r(i)}$, $t\in\bN$, where $\{r(t)\}_{t\in\bN}$ is an i.i.d. sequence of log-returns such that
\begin{equation}\label{eq:toy1}
\bP[r_i(t)=(\ln 1.5,\ln 0.5)]=\tfrac{1}{2},\quad \bP[r_i(t)=(\ln 0.6,\ln 1.8)]=\tfrac{1}{2}.
\end{equation}
Note that the sequence $\{r(t)\}_{t\in\bN}$ satisfies assumption \eqref{A.1}. Let us assume that the proportional transaction cost penalty is antisymmetric with 10\%/20\% costs, i.e let the penalty function be given by $d(x):=\langle c, [x]^+\rangle+ \langle h, [x]^-\rangle$, $x\in\bR^2$, for $c=(0.1,0.2)$ and $h=(0.2,0.1)$. Noting that for any $(\pi_1,\pi_2)\in\cS$ we get $(\pi_1,\pi_2)=(\pi_1,1-\pi_1)$ and using notation $\tilde v(\pi):=v(\pi,1-\pi)$ as well as $\tilde \pi :=(\pi,1-\pi)$  we can rewrite operator $T$ given in \eqref{eq:ex.T} as
\begin{equation}\label{eq:toy2}
T\tilde v(\pi)  =\sup_{\pi'\in [0,1]}
\left[\ln s(\tilde\pi,\tilde\pi')+\tfrac{1}{\gamma}\ln\bE\left(e^{\gamma \ln\langle \tilde \pi', w(1)\rangle+\gamma v(G(\tilde \pi',w(1)))}\right)\right]
\end{equation}
For the dynamics introduced in \eqref{eq:toy1}, we get that \eqref{eq:toy2} is equal to 
\[
\sup_{\pi'\in [0,1]}
\left[\ln s(\tilde\pi,\tilde\pi')-\tfrac{\ln 2}{\gamma}+\tfrac{1}{\gamma}\ln\left[e^{\gamma \ln (0.5+\pi')+\gamma\tilde v\left(\frac{3\pi'}{2\pi'+1}\right)}\\+e^{\gamma \ln (1.8-1.2\pi')+\gamma \tilde v\left(\frac{\pi'}{3-2\pi'} \right)}\right]\right],
\]
so that a simple univariate iterration scheme could be applied to recover the value of $T\tilde v$ given $\tilde v$. In fact, simple numerical calculations show that the value of iterated value function (centered with constant from \eqref{eq:constant} to increase calculation robustness) stabilizes relatively fast. This is illustrated in Figure~\ref{F:1}

\begin{figure}[htp!]
\begin{center}
\includegraphics[width=0.45\textwidth]{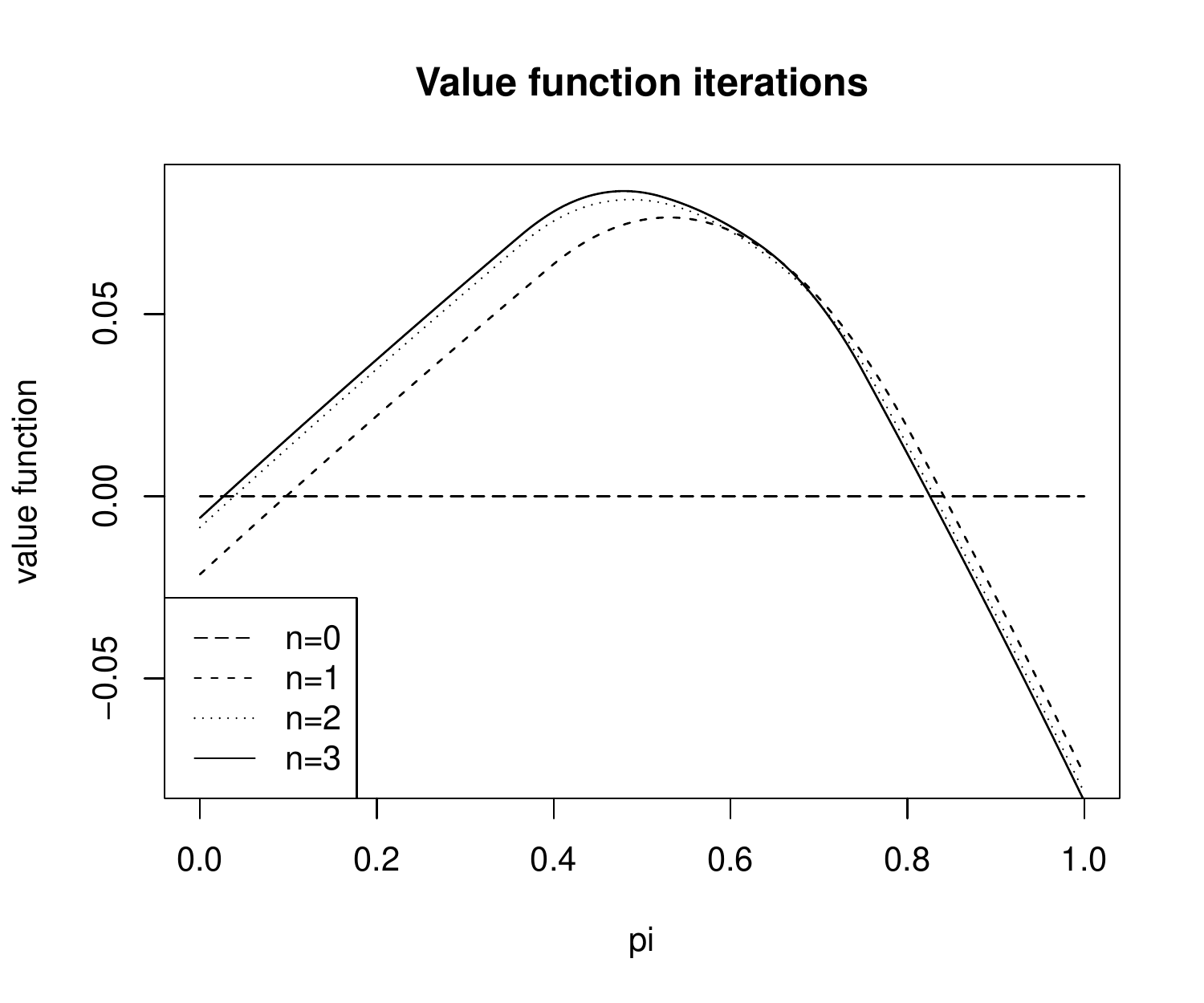}
\includegraphics[width=0.45\textwidth]{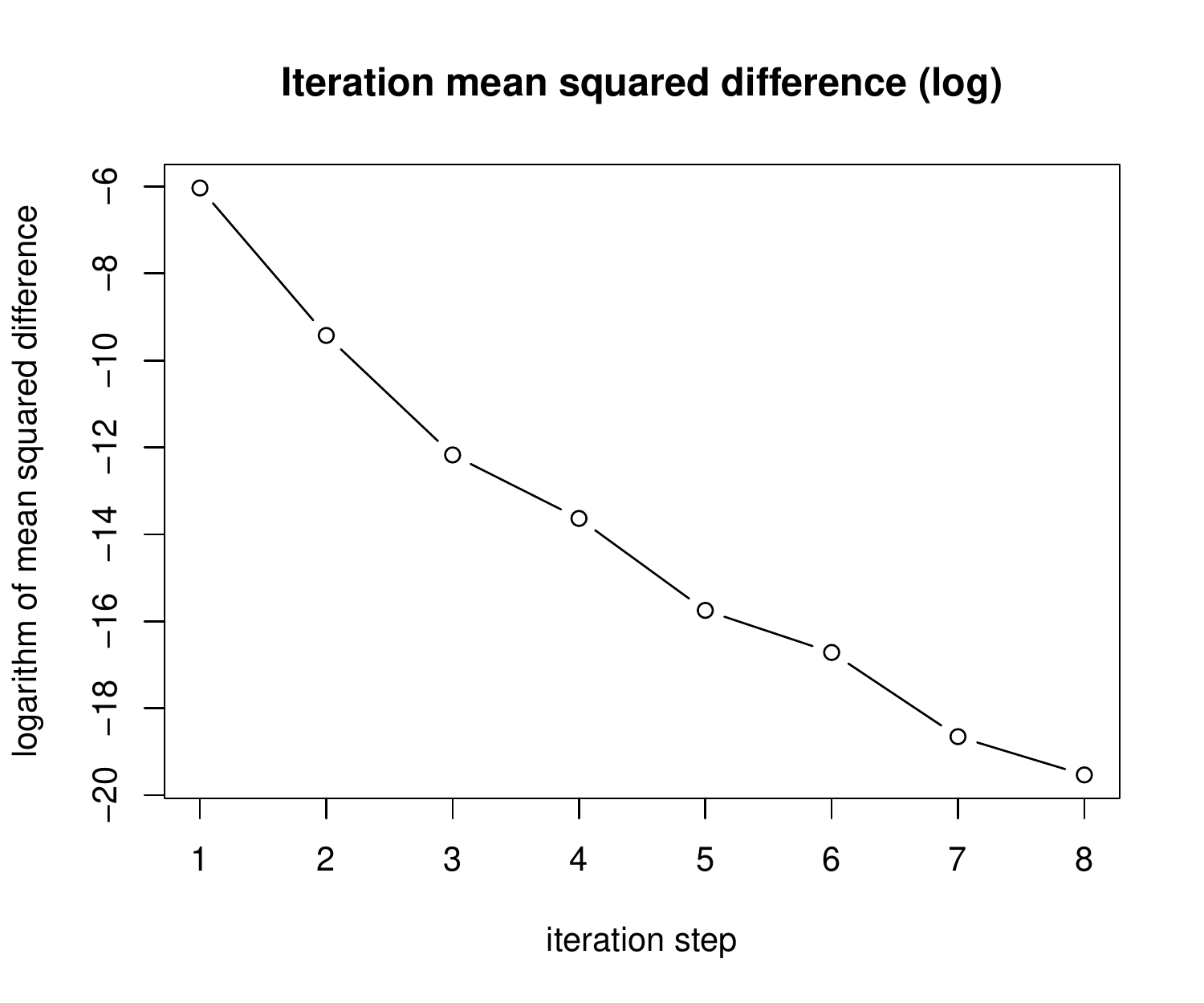}
\end{center}
\caption{The plot illustrates policy iteration convergence rates for Example~\ref{S:example1}. The left exhibit shows values of $T^n0$ as a function of the first weight, for $n=0,1,2,3$; note that we used centering constant from \eqref{eq:constant} to center the values of $(T^n0)$ to better present convergence in the span norm. The right plot shows logarithm of the mean squared difference of the subsequent (centered) iterations for $n=1,\ldots,9$.}
\label{F:1}
\end{figure}

Noting that the convergence rate is fast, we decided to set $n=8$ and use optimal strategy induced by $T^80$ in the optimisation. To illustrate the usefulness of risk-sensitive approach we study the asymptotic dynamics of the wealth function under various trading strategies for the dynamics introduced in \eqref{eq:toy1}. Namely, we consider four trading strategies: (1) Static buy-and-hold asset 1 strategy; (2) Static buy-and-hold asset 2 strategy; (3) Dynamic proportion strategy; (4) Dynamic strategy induced by risk-sensitive framework.

To compare the performance of trading strategies on a single long trajectory, we simulate a single realisation of $(r(t))_{t=1}^{T}$, for $T=5000$, with initial capital $W(0)=1$, Then, we apply the trading strategies to the dataset; the proportion in strategy (3) have been chosen in such a way that the final portfolio value is the highest, i.e. we decided to choose the trajectory-optimal proportion. The log-wealth evolution as well as trading strategy profiles could be found in Figure~\ref{F:2}. 

\begin{figure}[htp!]
\begin{center}
\includegraphics[width=0.45\textwidth]{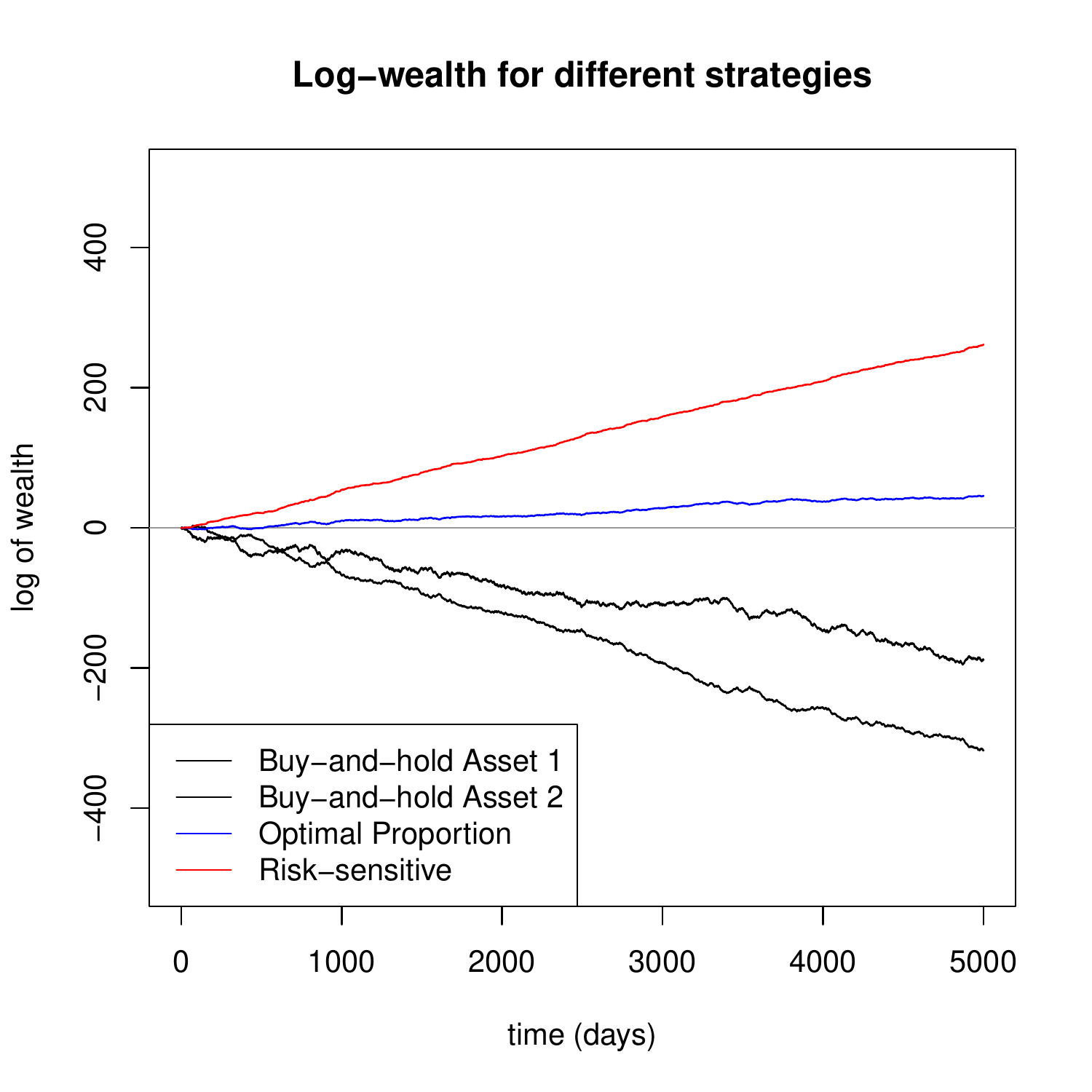}
\includegraphics[width=0.45\textwidth]{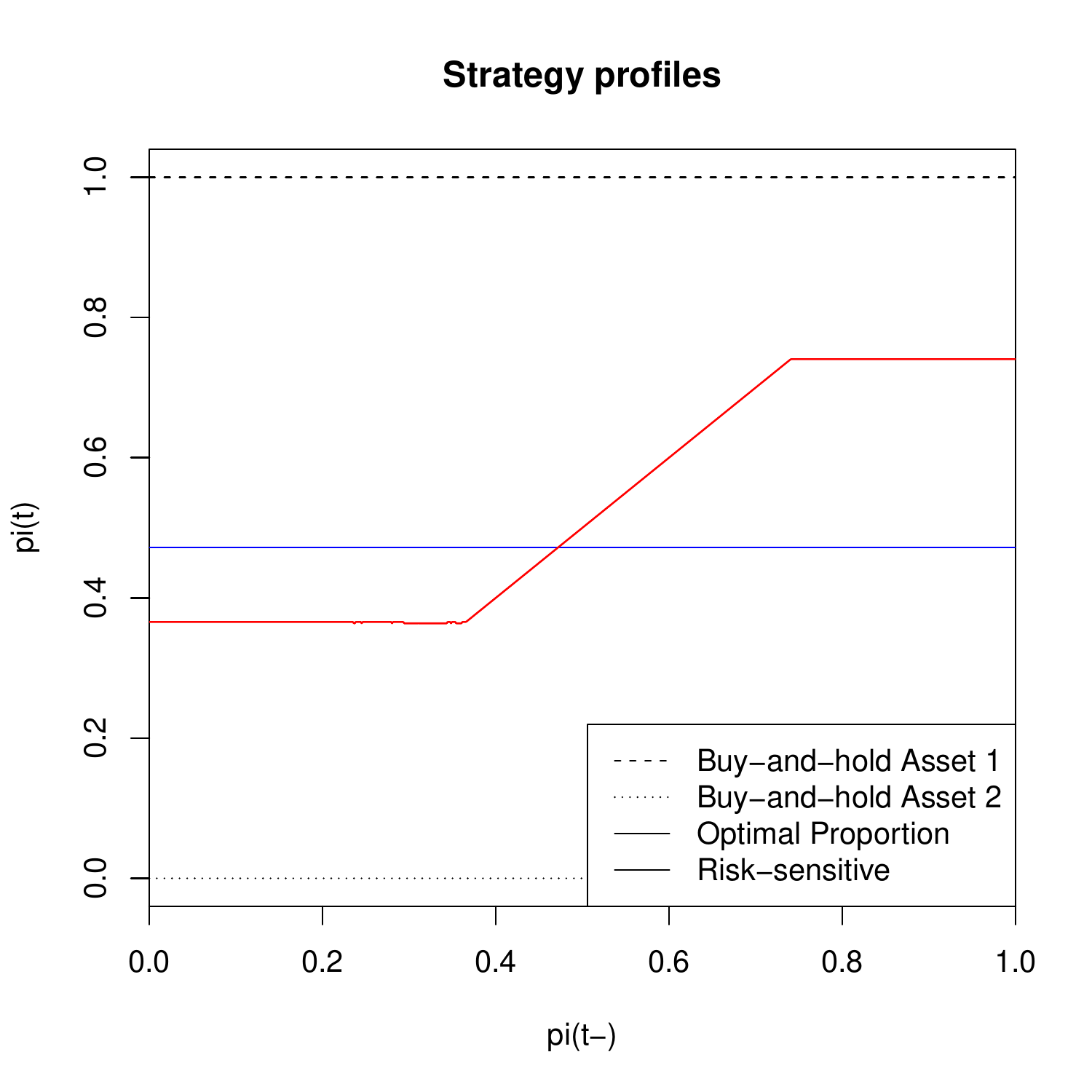}
\end{center}
\caption{In the left exhibit, we present the values of the log-wealth process $\ln W(t)$ for a single trajectory and $t=1,\ldots,5000$ under Example~\ref{S:example1} dynamics. One can see that the risk-sensitive strategy is outperforming all others and leads to stable time-growth. The right exhibit presents the trading profile of all strategies as a function of the first weight. As expected, the risk sensitive strategy is in fact a barrier strategy, i.e. no trading is made in we are inside a predefined interval (ca. $[0.38,0.75]$) and a push-back in initiated if we fall outside of it.}
\label{F:2}
\end{figure}

From the plot a couple of things can be deduced. First, the risk-sensitive strategy outperforms all other strategies and shows the benefits of dynamic rebalancing. In particular, while the static strategies lead to losses, the dynamic risk-sensitive strategy produces relatively stable gains. Second, as expected, the risk-sensitive induced strategy is a barrier-trading strategy. In other word, if for a given day $t$ the (proportion) value $\pi(t-)$ falls outside of some interval (ca. $[0.38,0.75]$), then the trading strategy push it into this interval with cost-efficient approach. On the other hand, if the value is inside the interval, then it is optimal to not impose any trading. Third, we see that it is advised to study dynamic trading strategies even in the most simplistic transaction-cost setting -- the performance improvement coming from applying risk-sensitive strategy to the problem looks material even when confronted with optimal proportion strategy.

Next, to fully understand the balance between the pay-off and the risk, we decided to calculate cumulative log-wealth for multiple trajectories and check their performance using assessment metrics. Namely, we consider the (time-normalised) mean, standard deviation, and entropy function $\mu^{\gamma}$; note the time-averaged entropy of log-wealth corresponds directly to the risk-sensitive objective function \eqref{eq:obj1}. For completeness, we also added values for entropy second-order Taylor expansion based on mean and variance which illustrates the link between risk-sensitive framework and mean-variance framework. Note that while mean-variance trading might lead to time-inconsistency, the risk-sensitive trading is time-consistent, see~\cite{BiePli2003,BieCheCia2021} for details. Also, to better understand the difference between risk-sensitive and risk-neutral frameworks, we decided to  include the results for strategy similar to (4) but with risk-sensitivity parameter set to $\gamma=-0.0005$, which approximates the risk-neutral setting, see \cite{DiMSte2006b}. With slight abuse of notation, we call it the {\it risk-neutral strategy}. The aggregated results are presented in Figure~\ref{F:3}.

\begin{figure}[htp!]
\begin{center}
\begin{tabular}{l|r|r|r|r}
 Strategy & Mean & Std & Mean+$\tfrac{\gamma}{2}\cdot$ Variance & Entropy ($\mu^{\gamma}$) \\\hline
Buy-and-hold asset 1& -0.053 & 0.029 & -0.106 & -0.101 \\
Buy-and-hold asset 2&  -0.053& 0.041 & -0.156 & -0.146
 \\
 Optimal proportion & 0.005 & 0.009 &  0.000& 0.000\\
 Risk-sensitive & 0.049 & {\bf 0.009} & {\bf 0.044} & {\bf 0.044}\\
 Risk-neutral & {\bf 0.050} & 0.012 & 0.042 & 0.043
\end{tabular}\label{F:3}
\caption{The table presents time-normalised performance metrics for trading strategies introduced in Example~\ref{S:example1}. The outputs are based on a strong Monte Carlo sample of size 20\,000 applied to log-wealth at time $T=250$. For completeness, we also added the results for risk-neutral dynamic strategy.}
\end{center}
\end{figure}

By looking at Figure~\ref{F:3} we see that the risk-sensitive trading strategy outperforms trading strategies (1)-(3). While the normalised log-mean for the risk-neutral  strategy is higher compared to risk-sensitive strategy (which is in fact expected as the log-mean reflects directly the objective function in the risk-neutral setting), it also results in higher variance. In other words, in the risk-sensitive case, the smaller mean (ca. 2\% decrease) is compensated by  risk decrease (ca. 25\% variance reduction) which seems like a plausible trade-off. This suggests that the risk-sensitive strategy might be the optimal choice among all considered strategies. Also, as expected, the risk-sensitive strategy has the highest entropy among all strategies.
\end{example}

\begin{example}[Gaussian noise with a drift]\label{S:example2}
In this example, we focus on a three dimensional correlated Gaussian noise with a positive drift. Let $d=3$, and let  $S(t)=S(0)\cdot\prod_{i=1}^{t} e^{r(i)}$, $t\in\bN$, where $\{r(t)\}_{t\in\bN}$ is an i.i.d. sequence of log-returns such that $r(t)\sim \mathcal{N}(\mu,\Sigma)$, where
\begin{equation}\label{eq:3dim1}
\mu:=0.001\cdot (2.5,1.5,2.0),\quad\textrm{and}\quad \Sigma:=0.0008\cdot
\begin{bmatrix}
 \phantom{-} 3.0 &-1.0 & -0.5 \\
-1.0 & \phantom{-}1.5 & \phantom{-}0.5\\
-0.5 & \phantom{-}0.5  & \phantom{-}2.0
\end{bmatrix}.
\end{equation}
Next, let $\gamma=-5$ and let the penalty function be given by $d(x)=\langle c, [x]^+\rangle+ \langle h, [x]^-\rangle$, $x\in\bR^3$, for $c=0.004\cdot (2.0,1.6,1.0)$ and $h=0.004\cdot (1.0,1.6,2.0)$. The parameters are fixed in such a way that we have both positive and negative correlation in assets; note the correlation coefficients for $r_i(t)$ are (ca.) equal to $\rho_{12}=-0.48$, $\rho_{13}=-0.20$ and $\rho_{23}=0.29$. Also, the  transaction costs reflect pay-off between return (mean) and risk (variance). The assumption \eqref{A.1} is satisfied for $\{r(t)\}_{t\in\bN}$ as the moment generating function for Gaussian random variables exists and assumption \eqref{A.2} is satisfied as the support of $\{r(t)\}_{t\in\bN}$ is full. 

First, as in Example~\ref{S:example1}, we approximate the solution to Bellman's equation. The approximation was made on a discrete $\pi$-grid of step-size 0.02. By analysing the consequent differences and the shape of the approximated value functions, we decided to set $n=5$, i.e use approximated value of $T^50$ to determine the trading strategy. The approximation results are illustrated in Figure~\ref{F:4}; note that for any $(\pi_1,\pi_2,\pi_3)\in\cS$ we get $\pi_3=1-(\pi_1+\pi_2)$, so it is sufficient to analyse a two-dimensional graphs.

\begin{figure}[htp!]
\begin{center}
\includegraphics[width=0.48\textwidth]{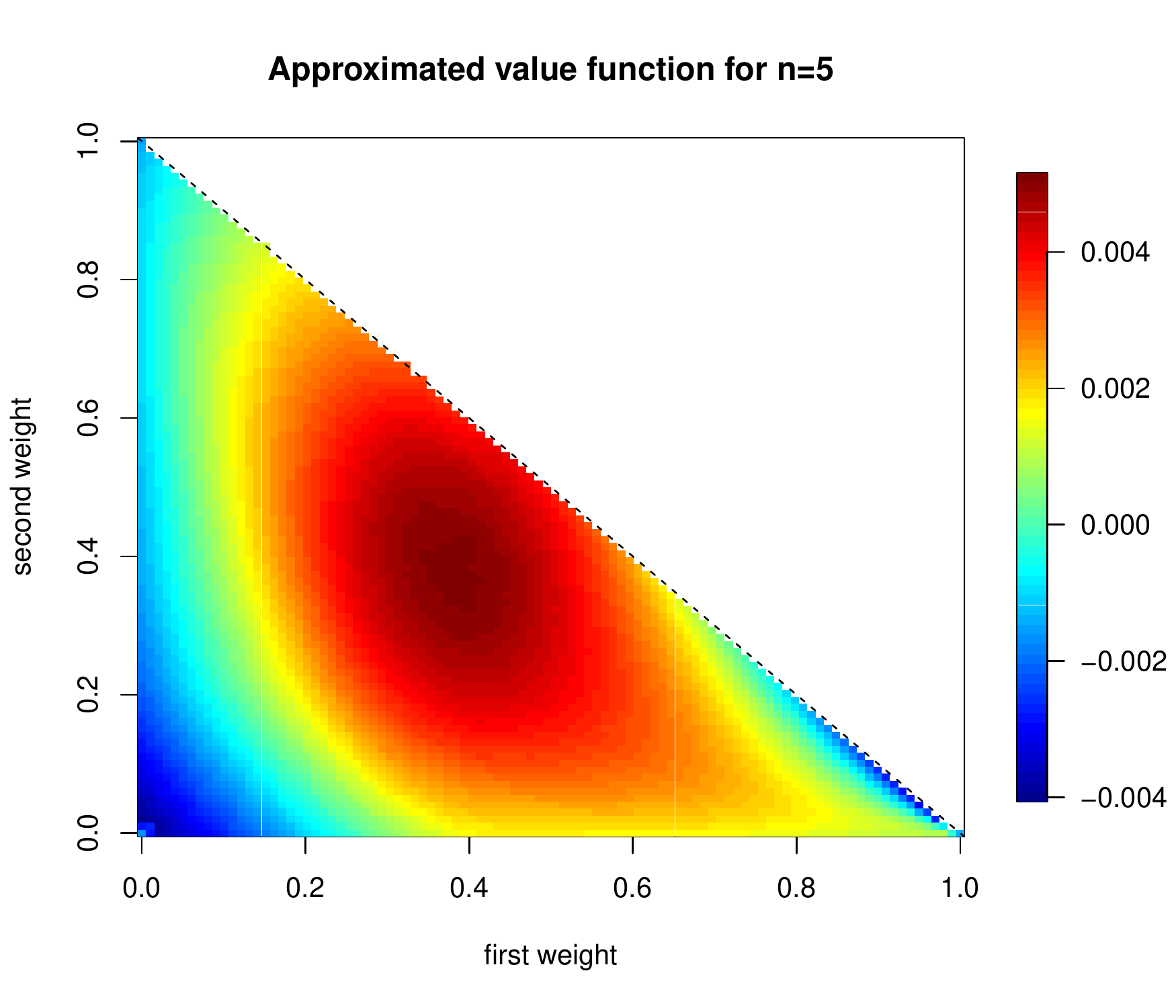}
\includegraphics[width=0.42\textwidth]{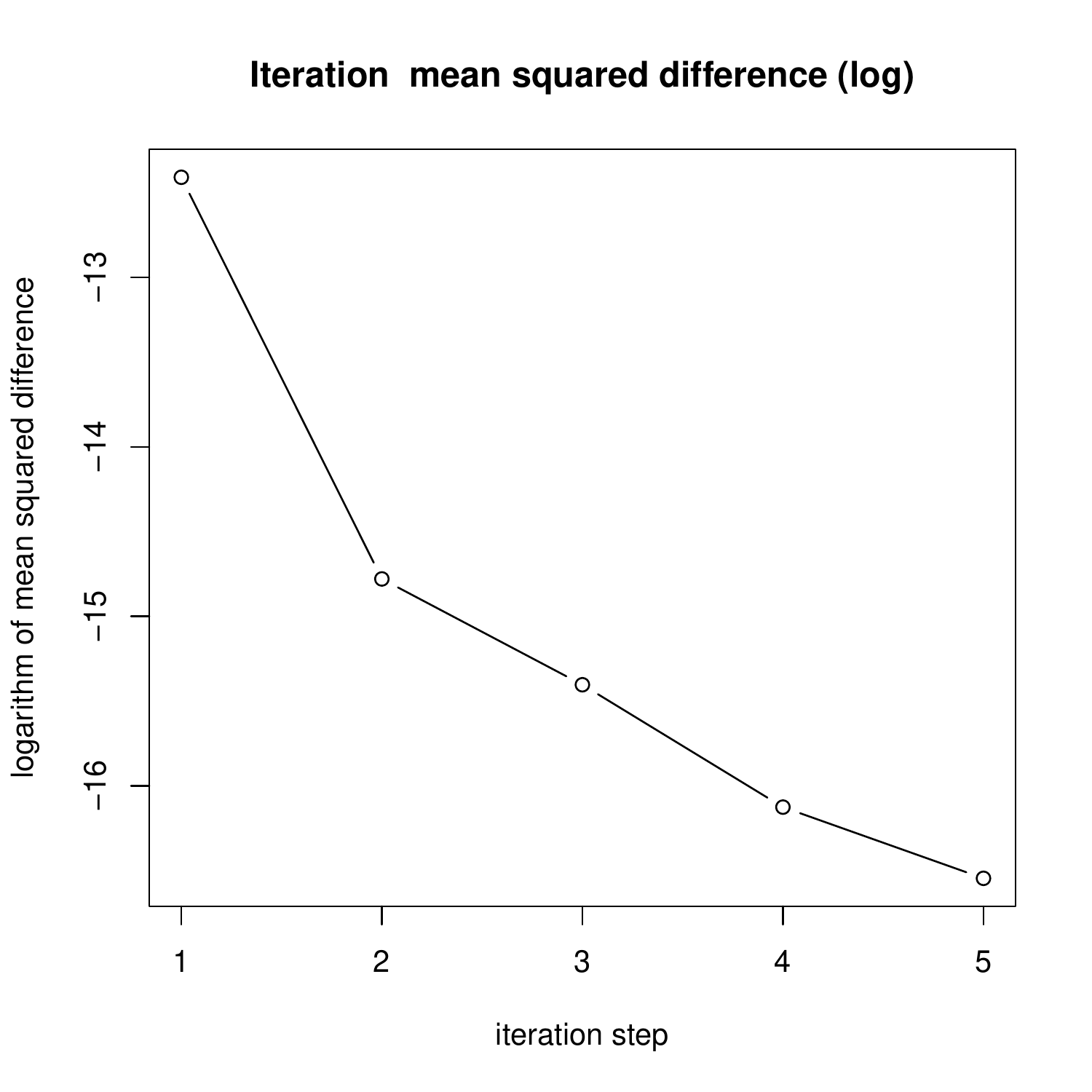}
\end{center}
\caption{The plot illustrates approximated value function (left) and iteration convergence rates (right) for Example~\ref{S:example2}. The left exhibit shows the approximated values of $T^50$ as a function of the first and second weight; note that we used centering constant from \eqref{eq:constant}. The right plot shows logarithm of the mean squared distance between the subsequent (centered) iterations for $n=1,\ldots,5$.}
\label{F:4}
\end{figure}

From Figure~\ref{F:4} we see that the approximated value function is regular with maximal value around point $(0.37, 0.40)$. Recalling that entropy utility function has a second order Taylor expansion $\mu^{\gamma}(X)\approx \bE[X] +\frac{\gamma}{2}\textrm{Var}[X]$, one would expect that this point (approximatelly) corresponds to an optimal allocation obtained using Markowitz portfolio optimization. This could be easily verified by solving a quadratic programming problem of the form
\begin{equation}\label{eq:ex2:markowitz}
\pi^{*}:={\arg\min}_{\pi\in\cS}(\pi^{T}\mu + \tfrac{\gamma}{2} \pi^{T} \Sigma \pi),
\end{equation}
for which we get $\pi^*\approx(0.3705357,0.4017857,0.2276786)$. 
Let us now confront the (approximated) trading strategy with other alternative strategies as in Example~\ref{S:example1}. Namely, we consider five trading strategies: (1) Static buy-and-hold asset 1 strategy; (2) Static buy-and-hold asset 2 strategy; (3) Static buy-and-hold asset 3 strategy; (4) Dynamic Markowitz proportion strategy; (5) Dynamic strategy induced by risk-sensitive framework. In strategy (4) we shift the allocation to the state induced by~\eqref{eq:ex2:markowitz}, i.e. we follow the optimal Markowitz strategy for risk aversion $\gamma$ under no transaction costs. While the full illustration of the trading strategy obtained via Bellman's equation approximation (as presented in Figure~\ref{F:2}) is problematic (it would require four-dimensional plot) we can analyse strategy structure by looking at {\it no-action points} as well as {\it shift points}, i.e. sets of $\pi$'s such that no trading is executed if we are in the state $\pi$ and sets of all $\pi$'s that are the target state for some pre-trading initial state. The trading results for an exemplary long single trajectory as well as simplified strategy profile is presented in Figure~\ref{F:5}.

\begin{figure}[htp!]
\begin{center}
\includegraphics[width=0.48\textwidth]{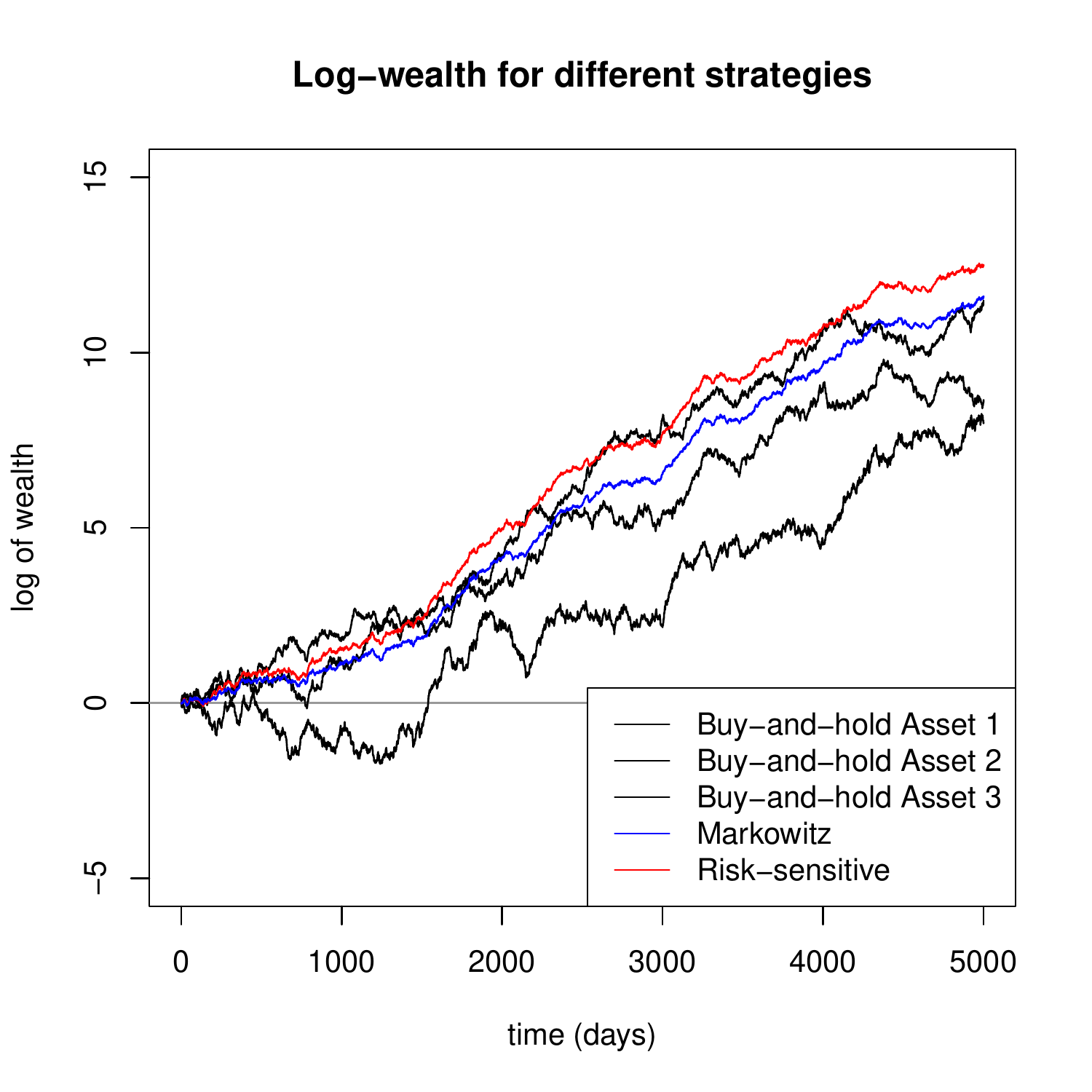}
\includegraphics[width=0.48\textwidth]{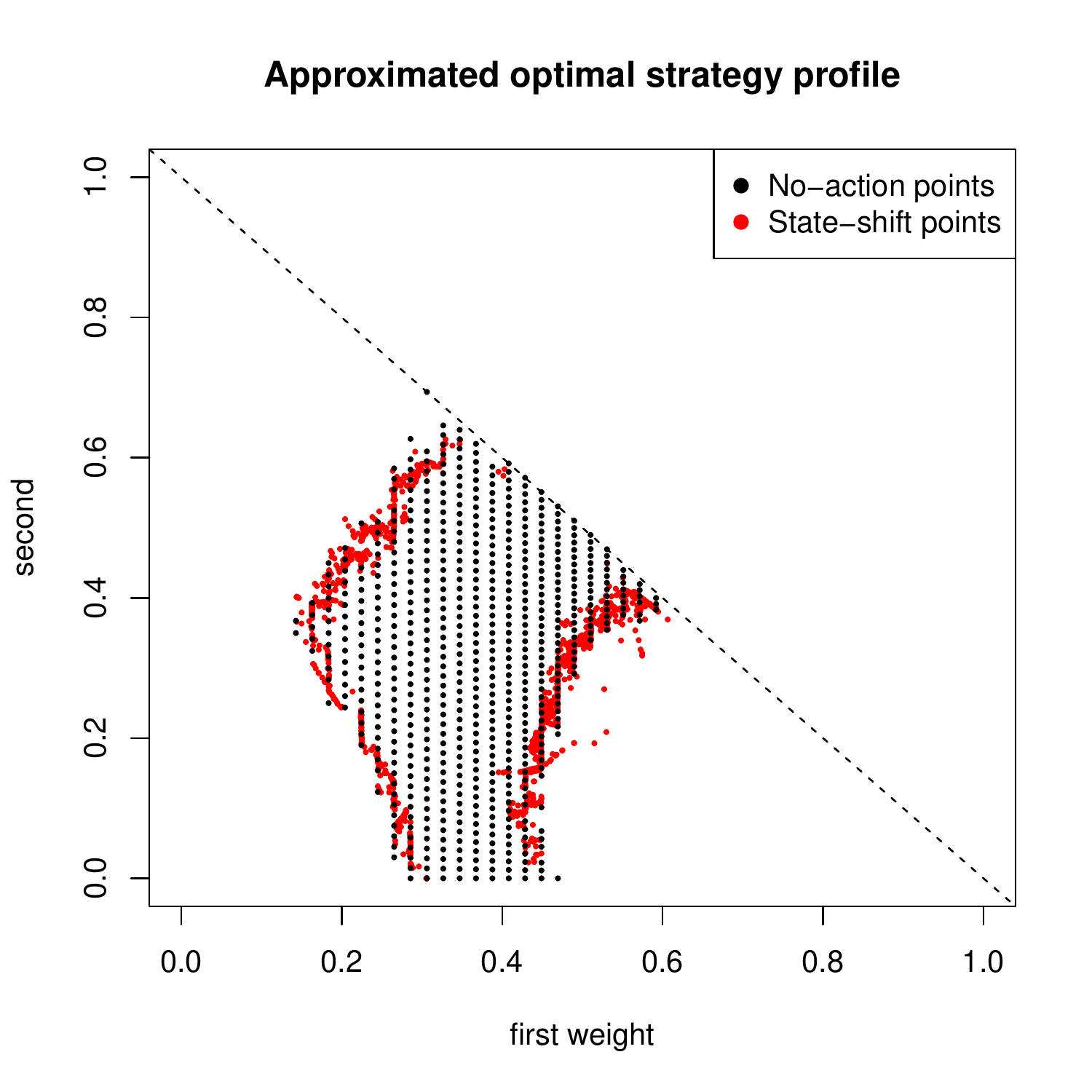}
\end{center}
\caption{In the left exhibit, we present the values of the log-wealth process $\ln W(t)$ for a single trajectory and $t=1,\ldots,5000$ under Example~\ref{S:example2} dynamics. One can see that the risk-sensitive strategy is outperforming all others. The right exhibit presents the structure of the risk-sensitive trading strategy -- no trading is executed if we are nearby Markowitz-induced optimal point and a push-back strategy is applied if we fall outside of the black-point area.}
\label{F:5}
\end{figure}
Next, we analyse trading performance by looking into (time-averaged) performance metrics introduced in Example~\ref{S:example1}. The aggregated results are presented in Figure~\ref{F:6}.

\begin{figure}[htp!]
\begin{center}
\begin{tabular}{l|r|r|r|r}
 Strategy & Mean & Std & Mean+$\tfrac{\gamma}{2}\cdot$ Variance & Entropy ($\mu^{\gamma}$)  \\\hline
Buy-and-hold asset 1& 0.0025 &0.0031  &-0.00358 &-0.00561\\
Buy-and-hold asset 2&  0.0015 & 0.0022 & -0.00149 &-0.00147\\
Buy-and-hold asset 3&  0.0020 & 0.0025 & -0.00203 &-0.00193\\
Markowitz proportion & 0.0024 & {\bf 0.0012} &  0.00144 & 0.00143\\
Risk-sensitive & {\bf 0.0026} & 0.0013 &  {\bf 0.00151} &  {\bf 0.00151}
\end{tabular}
\caption{The table presents time-normalised performance metrics for trading strategies introduced in Example~\ref{S:example2}. The outputs are based on a strong Monte Carlo sample of size 20\,000 applied to log-wealth at time $T=250$.}
\label{F:6}
\end{center}
\end{figure}

From Figure~\ref{F:6} we see that risk-sensitive trading strategy is outperforming all other strategies and has the highest entropy, as expected. While the variance for Markowitz strategy is slightly smaller, the Markowitz allocation leads to smaller mean -- the payoff between the two is better for risk-sensitive strategy as could be seen by looking into both {\it Entropy} and ${\textrm Mean}+\tfrac{\gamma}{2}\textrm{Variance}$ performance criterions. To better understand the difference between risk-sensitive strategy and Markowitz strategy it is best to look into trading intensity. While the Markowitz strategy has homogeneous trading intensity (as expected, as the strategy should always push the allocation back to the fixed point), the risk-sensitive strategy shows more intense trading on rare occasions, i.e. when the process falls outside of the zone presented in Figure~\ref{F:5}. In fact, while the trading for Markowitz strategy was initiated for almost all of the considered days, the trading for risk-sensitive strategy was initiated in only 176 days (ca. 3.5\% of sample size). Moreover, the aggregated trading intensity for risk-sensitive strategy, measured e.g. by cumulative product of capital decays, is much smaller. This shows that intense re-balancing could in fact negatively impact performance, especially in the transaction cost regime.


\end{example}

\nocite{MacThoZie2011}
\nocite{Ste2005}
\nocite{Ste2009}
\nocite{DunDunSte2011}
\nocite{ChrIrlLud2017}
\nocite{Ste2011b}

\bibliographystyle{siamplain}
\bibliography{RSC_bibliografia}

\end{document}